\documentclass[12pt]{article}%

\usepackage{amssymb}
\usepackage{amsfonts}
\usepackage{amsmath}
\usepackage[nohead]{geometry}
\usepackage[singlespacing]{setspace}
\usepackage[bottom]{footmisc}
\usepackage{indentfirst}
\usepackage{endnotes}
\usepackage{graphicx}
\usepackage{epstopdf}
\usepackage{rotating}
\usepackage{verbatim}
\usepackage{setspace}
\usepackage{multirow}
\usepackage{latexsym}
\usepackage{color}
\usepackage{float}

\usepackage{amsthm}
\usepackage{makeidx}
\usepackage{fancyhdr}
\usepackage{type1cm}
\usepackage{bbm}
\usepackage{enumitem}
\usepackage{algorithm}
\usepackage{algpseudocode}

 \usepackage[authoryear]{natbib}

\newcommand{\bb}{\mathrm{\bf b}}
\newcommand{\bff}{\mathrm{\bf f}}

\newcommand{\bA}{\mathrm{\bf A}}
\newcommand{\ba}{\mathrm{\bf a}}
\newcommand{\bc}{\mathrm{\bf c}}
\newcommand{\bd}{\mathrm{\bf d}}
\newcommand{\bu}{\mathrm{\bf u}}
\newcommand{\bv}{\mathrm{\bf v}}

\newcommand{\bB}{\mathrm{\bf B}}

\newcommand{\bC}{\mathrm{\bf C}}

\newcommand{\bD}{\mathrm{\bf D}}
\newcommand{\bE}{\mathrm{\bf E}}
\newcommand{\bF}{\mathrm{\bf F}}
\newcommand{\bW}{\mathrm{\bf W}}
\newcommand{\bJ}{\mathrm{\bf J}}
\newcommand{\bG}{\mathrm{\bf G}}

\newcommand{\bH}{\mathrm{\bf H}}
\newcommand{\bI}{\mathrm{\bf I}}
\newcommand{\bg}{\mathrm{\bf g}}
\newcommand{\bP}{\mathrm{\bf P}}

\newcommand{\bQ}{\mathrm{\bf Q}}
\newcommand{\bK}{\mathrm{\bf K}}
\newcommand{\bR}{\mathrm{\bf R}}

\newcommand{\bU}{\mathrm{\bf U}}
\newcommand{\bX}{\mathrm{\bf X}}
\newcommand{\bY}{\mathrm{\bf Y}}

\newcommand{\bxi}{\mbox{\boldmath $\xi$}}
\newcommand{\bXi}{\mbox{\boldmath $\Xi$}}
\newcommand{\bPi}{\mbox{\boldmath $\Pi$}}

\newcommand{\bLambda}{\mbox{\boldmath $\Lambda$}}
\newcommand{\blambda}{\mbox{\boldmath $\lambda$}}
\newcommand{\bgamma}{\mbox{\boldmath $\gamma$}}
\newcommand{\bGamma}{\mbox{\boldmath $\Gamma$}}

\newcommand{\bSigma}{\mbox{\boldmath $\Sigma$}}
\newcommand{\bOmega}{\mbox{\boldmath $\Omega$}}

\newcommand{\bPhi}{\mbox{\boldmath $\Phi$}}
\newcommand{\bepsilon}{\mbox{\boldmath $\epsilon$}}

\newcommand{\bDelta}{\mbox{\boldmath $\Delta$}}

\newcommand{\hF}{\widehat \bF}

\newcommand{\tF}{\widetilde \bF}
\newcommand{\tG}{\widetilde \bG}
\newcommand{\tU}{\widetilde \bU}

\newcommand{\cov}{\mathrm{cov}}
\newcommand{\Var}{\mathrm{Var}}

\newcommand{\tr}{\mathrm{tr}}

\newcommand{\var}{\mathrm{var}}
\newcommand{\beq}{\begin{eqnarray*}}
\newcommand{\eeq}{\end{eqnarray*}}

\numberwithin{equation}{section}
\theoremstyle{plain}
\newtheorem{thm}{Theorem}[section]

\newtheorem{lem}{Lemma}[section]

\newtheorem{prop}{Proposition}[section]
\newtheorem{assum}{Assumption}[section]
\theoremstyle{definition}

\usepackage{textcomp}
\usepackage{xr}
\makeatletter
\def\@biblabel#1{\hspace*{-\labelsep}}
\makeatother
\begin{document}

 \title{
 \bf Heterogeneity Adjustment with
 Applications to Graphical Model Inference
 }
\author{Jianqing Fan\thanks{Address: Department of ORFE, Sherrerd Hall, Princeton University, Princeton, NJ 08544, USA, e-mail: \textit{jqfan@princeton.edu},
\textit{hanliu@princeton.edu}, \textit{weichenw@princeton.edu},  \textit{ziweiz@princeton.edu}. This project was supported by National Science Foundation grants DMS-1206464, DMS-1406266 and 2R01-GM072611-12.}$\;$, Han Liu, Weichen Wang, and Ziwei Zhu
\medskip\\{\normalsize Department of Operations Research and Financial Engineering,  Princeton University}}

\date{}

\maketitle

\sloppy

\onehalfspacing

\begin{abstract}
Heterogeneity is an unwanted variation when analyzing aggregated  datasets from multiple sources. Though different methods have been proposed for heterogeneity adjustment, no systematic theory exists to justify these methods. In this work, we propose a  generic  framework named ALPHA (short for \underline{A}daptive \underline{L}ow-rank \underline{P}rincipal \underline{H}eterogeneity \underline{A}djustment) to model, estimate, and adjust heterogeneity from the original data.  Once the heterogeneity is adjusted, we are able to remove the biases of batch effects and to enhance the inferential power by aggregating the  homogeneous residuals from multiple sources. Under a pervasive assumption that the latent heterogeneity factors simultaneously affect a large fraction of observed variables, we provide  a rigorous theory to justify the proposed framework.
Our framework also allows the incorporation of informative covariates and appeals to  the `Bless of Dimensionality'.  As an illustrative application of this generic framework, we consider a problem of estimating high-dimensional precision matrix for  graphical model inference based on multiple datasets.  We also provide thorough numerical studies on both synthetic datasets and a brain imaging dataset to demonstrate the efficacy of the developed theory and methods.
	
\end{abstract}



\onehalfspacing

\section{Introduction}

Aggregating and analyzing heterogeneous data is one of the most fundamental challenges in scientific data analysis. In particular, the  intrinsic heterogeneity across multiple data sources  violates the ideal  `independent and identically distributed' sampling assumption and may produce misleading results if it is ignored. For example, in genomics, data heterogeneity is ubiquitous and referred to as either `batch effect' or `lab effect'. Microarray gene expression data obtained from different labs at different processing dates may contain systematic variability.  More specifically, \cite{LSB10} analyzed a microarray data from a bladder cancer study and showed that the gene expressions vary significantly across different batches even  after data normalization. Furthermore, \cite{LSt07} pointed out that heterogeneity across multiple data sources may be caused by unobserved factors that have confounding effects on the variables of interest, generating spurious signals. In finance, it is also known that asset returns are driven by varying market regimes and economy status,  which can be regarded as a temporal batch effect.  Later in this paper, we will  use a brain imaging dataset to show similar heterogeneity effect. Therefore, to properly analyze data aggregated from multiple sources, we need to carefully model and adjust the heterogeneity effect.

Modeling and estimating heterogeneity effect is challenging for two reasons. (i) Typically, we can only access a limited number of samples from an individual group, given the high cost of biological experiment, technological constraint or fast economy regime switching. (ii) The dimensionality can be much larger than the total aggregated number of samples. The past decade has witnessed the development of many methods for adjusting batch effect in high throughput genomics data.
See, for example, \cite{SSH08, ABB00,LSt07, JLR07}.
Though progresses have been made, most of the aforementioned papers focus on the practical side and none of them has a systematic theoretical justification. In fact,  most of these methods are developed in a case-by-case fashion and are only applicable to certain problem domains. Thus, there is still a gap that exists between practice and theories.

To bridge this gap, we propose a generic theoretical framework to model, estimate, and adjust heterogeneity across multiple datasets. Formally, we assume the data come from $m$ different sources: the $i^{th}$ data source contributes $n_i$ samples, each having $p$ measurements such as gene expressions of an individual or stock returns of a day. To explicitly model heterogeneity, we assume that the batch-specific latent factors $\bff_t^{i}$ influence the observed data $X_{jt}^{i}$ in batch $i$ ($j$ indexes variables; $t$ indexes samples) as in the approximate factor model:
\begin{equation} \label{eq1.1}
    X_{jt}^{i} = {\blambda_j^{i}}' \bff_t^{i} + u_{jt}^{i}, \;\; 1\le j \le p, 1 \le t \le n_i, 1\le i \le m,
\end{equation}
where $\blambda_j^{i}$ is unknown factor loading for variable $j$ and $u_{jt}^i$ is true uncorrupted signals. The linear term ${\blambda_j^{i}}' \bff_t^{i}$ models the heterogeneity effect. We assume that $\bff_t^i$ is independent of $u_{jt}^i$ and $\bu_t^i = (u_{1t}, \dots, u_{pt})'$ shares the same common distribution with mean ${\bf 0}$ and covariance $\bSigma_{p\times p}$ across all data sources.
In the matrix-form model, \eqref{eq1.1} can be written as
\begin{equation} \label{eq1.2}
    \bX^i = \bLambda^i {\bF^i}' + \bU^i,
\end{equation}
where $\bX^i$ is a $p \times n_i$ data matrix in the $i^{th}$ batch, $\bLambda^i$ is a $p \times K^i$ factor loading matrix with $\blambda_j^i$ in the $j^{th}$ row, $\bF^i$ is an $n_i \times K^i$ factor matrix and $\bU^i$ is a signal matrix of dimension $p \times n_i$. Here, we allow the number of latent factors $K^i$ to depend on batch $i$.

To see how model \eqref{eq1.2} models the heterogeneity, we assume $\bff_t^i \sim N({\bf 0}, \bI)$ and $\bu_t^i \sim N({\bf 0}, \bSigma)$.  Then, the $t^{th}$ sample $\bX_t^i$, which is the $t^{th}$ column of $\bX^i$, follows
\begin{equation} \label{eq1.3}
    \bX_t^i \sim N({\bf 0}, \bLambda^i {\bLambda^i}' + \bSigma).
\end{equation}
Therefore, the heterogeneity effect is modeled as a low rank component $ \bLambda^i {\bLambda^i}'$ of the population covariance matrix of $\bX_t^i$. Later, we will show that, under a pervasive assumption, the heterogeneity component can be estimated by directly applying principal component analysis (PCA) or Projected-PCA, which is more accurate when there are sufficiently informative covariates $\bW^i$ \citep{FLW14}. Let $ \widehat{\bLambda^i} {\widehat{\bF^i}}' $ be the estimated heterogeneity component.  We denote $\widehat{\bU^i} = \bX_t^i -\widehat{\bLambda^i} {\widehat{\bF^i}}' $ to be the heterogeneity adjusted signal, which can be treated as homogeneous across different datasets and thus can be combined together for downstream statistical analysis. This whole framework of heterogeneity adjustment is termed ALPHA (short for \underline{A}daptive \underline{L}ow-rank \underline{P}rincipal \underline{H}eterogeneity \underline{A}djustment) and is schematically shown in Figure \ref{alpha}.

\begin{figure}[htp]
	\centering
	\includegraphics[scale=.21]{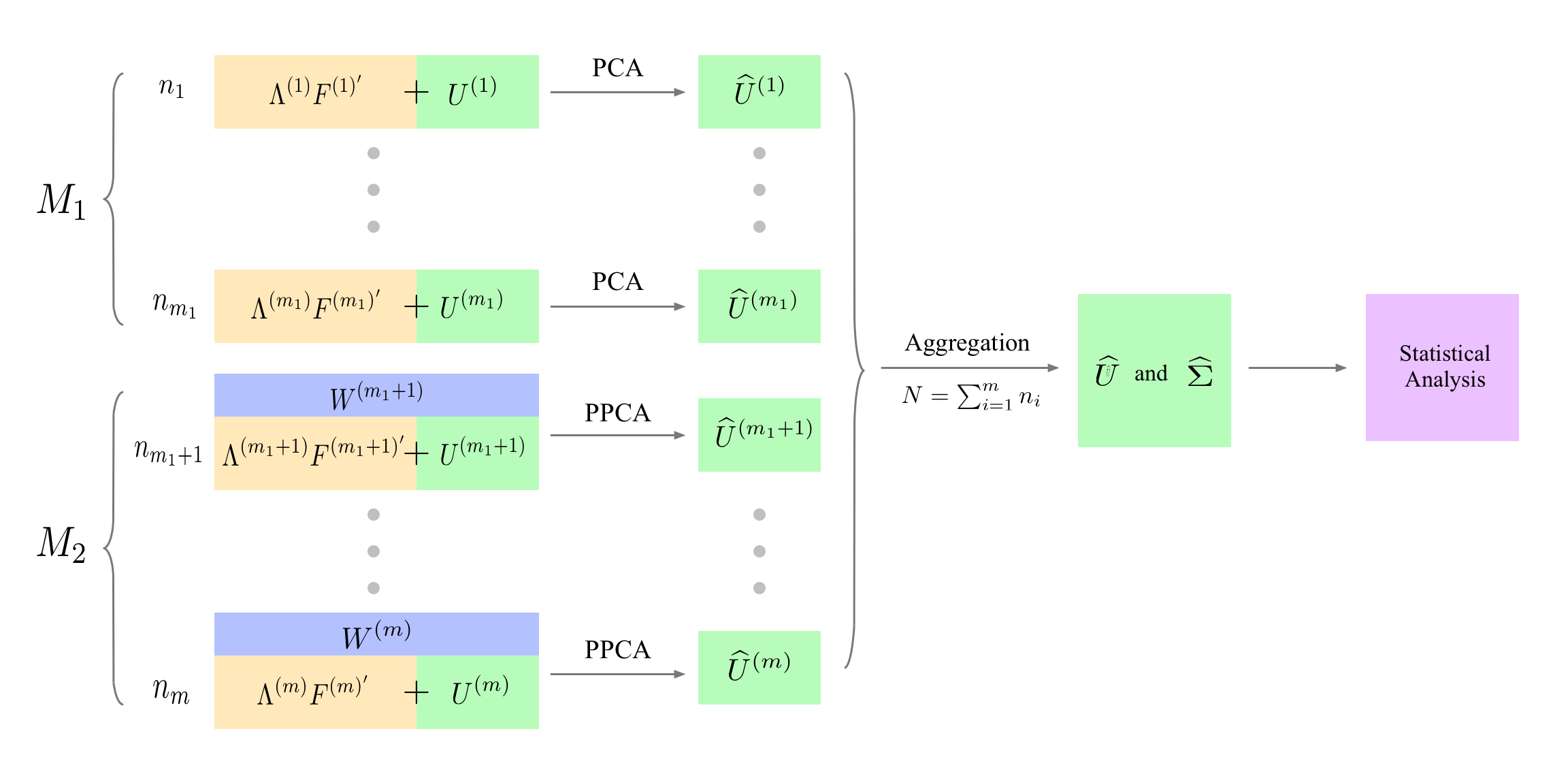}
	\caption{\small Schematic illustration of ALPHA: Depending whether we can find some sufficiently informative covariates $\bW$, we implement principal component analysis (PCA) or Projected-PCA (PPCA) methods (labeled respectively $M_1$ and $M_2$) to remove the heterogeneity effects $\bLambda \bF'$ for each batch of data.  This decision was made adaptively by a test statistic. After removing the unwanted variations, the homogeneous data $\{\bU^{(i)}\}_{i=1}^m$ are aggregrated
for further analysis.}
	\label{alpha}
\end{figure}

The proposed ALPHA framework is fully generic and applicable to almost all kinds of multivariate analysis of the combined, heterogeneity adjusted datasets.  As an illustrative example,  in this paper, we focus on the problem of Gaussian graphical model inference based on multiple datasets.  It is a powerful tool to explore complex dependence structure among variables  $\bX = (X_1,\dots,X_p)' \sim N({\bf 0}, \bSigma)$. The sparsity pattern of the precision matrix $\bOmega = \bSigma^{-1}$ encodes the information of an undirected graph $G = (V, E)$ where $V$ consists of $p$ vertices corresponding to $p$ variables in $\bX$ and $E$ describes their dependence relationship. To be specific, $V_i$ and $V_j$ are linked by an edge if and only if $\Omega_{ij} \ne 0$, meaning that $X_i$ and $X_j$ are dependent conditioning on the rest variables. For heterogeneous data across $m$ data sources, we need to first adjust for heterogeneity using the ALPHA framework. The idea of covariate-adjusted precision matrix estimation has been studied by \cite{CLL12}, but the factor model they used assumes observed factors and no heterogeneity issue, i.e., $m=1$.

A significant amount of literature has focused on the estimation of the precision matrix $\bOmega$ for graphical models for homogeneous data. 
\cite{YLi07}, \cite{BEd08}, \cite{FHT08} developed the Graphical Lasso method using the $L_1$ penalty and \cite{LFa09} and \cite{SPZ12} used a non-convex penalty. Furthermore, \cite{RWR11} and \cite{LWo13} studied the theoretical properties under different assumptions.
Estimating $\bOmega$ can be equivalently reformulated as a set of node-wise sparse linear regression that utilizes Lasso or Danzig selector for each node \citep{MBu06,Yua10,CLL11}.
To relax the assumption of Gaussian data,  \cite{LLW09} and \cite{LiuHanZha12} extend the graphical model to the case of semiparametric Gaussian copula and transelliptical family. Under the ALPHA framework, the adjusted data $\widehat {\bU^i}$ can be combined to construct an estimator for the  inverse matrix $\bOmega$ by the above   methods.

The rest of the paper is organized as follows. Section \ref{sec2} lays out a basic problem setup and necessary assumptions. We model the heterogeneity by a semiparametric factor model. Section \ref{sec3} introduces the ALPHA methodology for heterogeneity adjustment. Two main methods of PCA and Projected-PCA will be introduced for adjusting the factor effects under different settings. A guiding rule of thumb is also proposed to determine which method is more appropriate in the real data analysis. The heterogeneity-adjusted data will be combined to provide valid graph estimation in Section \ref{sec4}. The CLIME method of \cite{CLL11} is used for estimating precision matrix, although other related methods are also applicable. Some synthetic simulations and a real dataset are analyzed to demonstrate the proposed framework in Section \ref{sec5}. Section \ref{sec6} contains some further discussions. All the proofs are relegated to the Appendix.

\section{Problem Setup} \label{sec2}

To more efficiently use the external covariate information in removing heterogeneity effect,  we first present a semiparametric factor model.  Then, based on whether the collected external covariates have explaining power on factor loadings, we discuss two different regimes where PCA or Projected-PCA (PPCA) should be used. We will state the conditions under which these methods can be formally justified.

\subsection{Semiparametric factor model}

We assume that for subgroup $i$, we have $d$ external covariates $\bW_j^i = (W_{j1}^i, \dots, W_{jd}^i)'$ for variable $j$. In stock returns, these can be attributes of a firm; in brain imaging, these can be the physical locations of voxels.  We assume that these covariates have some explanatory power on the loading parameters $\blambda^{i}_{j}$ in (\ref{eq1.1}) so that it can be further modeled as $\blambda_j^i = \bg^i(\bW_j^i) + \bgamma_j^i$, where $\bg^i(\cdot)$ is the external covariate effects on $\blambda^{i}_{j}$ and $\bgamma^i_j$ is the part that can not be explained by the covariates.  Thus,
model (\ref{eq1.1}) can be written as
\begin{equation} \label{eq2.1}
X^{i}_{jt} = {\blambda^{i}_{j}}' \bff_{t}^i+u^{i}_{jt} =(\bg^i(\bW_j^i) + \bgamma_j^i)' \bff_t^{i} + u_{jt}^{i}.
\end{equation}
Model (\ref{eq2.1}) does not put much restriction.  If $\bW_j^i$ is not informative (i.e., ${\blambda^{i}_{j}}$ does not depend on $\bW_j^i$), then
$\bg^i(\cdot) = 0$, the model reduces to a regular factor model.  In a matrix form,  model \eqref{eq2.1} can be written as
\begin{equation} \label{eq2.2}
\bX^i = \bLambda^i {\bF^i}' + \bU^i\;\; \text{where} \;\; \bLambda^i = \bG^i(\bW^i) + \bGamma^i, \;\; 1\le i \le m.
\end{equation}
In \eqref{eq2.2},  $\bG^i(\bW^i)$ and  $\bGamma^i$ are $p \times K^i$ component  matrices of the factor loading $\bLambda^i$. More specifically, $g_k^i(\bW_j^i)$ and $\gamma_{jk}$ are the $(j,k)^{th}$ element of $\bG^i(\bW^i)$ and $\bGamma^i$ respectively. Expression (\ref{eq2.2}) suggests that the observed data can be decomposed into a low-rank heterogeneity term $\bLambda^i {\bF^i}'$  and a homogeneous signal  term $ \bU^i$.

Letting $\bu_t^i$ be the $t^{th}$ column of $\bU^i$, we assume all $\bu_t^i$'s  share the same distribution for any $t \le n_i$ and for all subgroups $i\le m$ with $\mathbb E[\bu_{t}^i] = {\bf 0}, \Var(\bu_{t}^i) = \bSigma$. Our goal is to recover $\bU^i$ from the observation $\bX^i$ and combine all the estimated $\bU^i$'s together to enhance the inferential power of $\bSigma$ or $\bOmega=\bSigma^{-1}$.

There has been a large literature on factor models in econometrics \citep{Bai03, BN13, FanLiaMin13, SW02}, machine learning \citep{CaiMaWu13, NegWai11, CanLiMaWri11} and  random matrix theories \citep{JohLu09, Pau07, SSZM13, FanWan15}. We refer the interested readers to those relevant papers and the references therein. However, none of these models incorporate the external covariate information. The semiparametric factor model \eqref{eq2.1} was first proposed  by \cite{CL07} and further investigated by \cite{CMO12, FLW14}. Using sufficiently informative external covariates,
we are able to
more accurately estimate the factors and loadings, and hence yield better heterogeneous adjustment.

\subsection{Modeling assumptions and general methodology}

In this subsection, we explicitly list all the required modeling assumptions. We start with an introduction of the data generating processes.

\begin{assum}[Data Generating Process] \label{ass2.2}
(i) $n_i^{-1} {\bF^i}'\bF^i = \bI$. \\
(ii) $\{\bu_t^i\}_{t \le n_i, i \le m}$ are independent within and between subgroups. $\bu_t^i$'s are identically sub-Gaussian distributed with mean zero and variance $\bSigma$ across all subgroups and are independent of $\{\bW_j^i, \bff_t^i\}$. $\{\bff_t^i\}_{t \le n_i}$ is a stationary process, but with arbitrary temporal dependency.
\\
(iii) There exists a constant $C_0 > 0$ such that
$\|\bSigma\|_2 \le C_0$. \\
(iv) The tail of the factors is sub-Gaussian, i.e.,
 $\exists C_1, C_2 > 0$ such that for $k \le K^i, t \le n_i$,
$$P(|f_{tk}^i|>t)\leq C_1\exp(-C_2 t^2).$$
\end{assum}

The above set of assumptions are commonly used in the literature, see \cite{BN13, FLW14}. We omit detailed discussions here.


Based on whether the external covariates are informative, we specify two regimes, each of which requires some additional technical conditions.

\subsubsection{Regime 1: External covariates are not informative}
For the case that $\bG^i(\bW^i) ={\bf 0}$, the external covariates do not have explanatory power on the factor loadings $\bLambda^i$ and model (\ref{eq2.2}) reduces to the traditional factor model, extensively studied in econometrics \citep{Bai03, SW02, Ona12}.
PCA will be employed in Section \ref{sec3.1} to estimate the heterogeneous effect.
It requires the following assumptions.

\begin{assum} \label{ass2.1.1}
(i) (Pervasiveness) There are two positive constants $c_{\min}$ and $c_{\max}$ so that
$$
c_{\min}<\lambda_{\min}(p^{-1}{\bLambda^i}'\bLambda^i)<\lambda_{\max}(p^{-1}{\bLambda^i}'\bLambda^i)<c_{\max}, \quad a.s. \quad \forall i.\\
$$
(ii) $\max_{k\leq K^i, j\leq p} |\lambda_{jk}^i| = O_P(\sqrt{\log p})$.
\end{assum}

The first condition is common and essential in the factor model literature (e.g., \cite{SW02}). It requires the factors  to be strong enough such that the covariance matrix $\bLambda^i \cov(\bff_t^i) \bLambda^i +\bSigma$ has spiked eigenvalues. This is trivially true if $\{\blambda_j^i\}_{j=1}^p$'s can be regarded as random samples from a population with nondegenerate sample covariance matrix \citep{FanLiaMin13}.  The second condition is technical, and a relaxation of bounded requirement in the literature \citep{FanLiaMin13, BN13}.

\subsubsection{Regime 2: External covariates are informative}
When covariates are informative, we will employ the PPCA \citep{FLW14} to better estimate the heterogeneous effect.  It requires the following assumptions.

\begin{assum} \label{ass2.1.2}
(i) (Pervasiveness) There are two positive constants $c_{\min}$ and $c_{\max}$ so that
$$
c_{\min}<\lambda_{\min}(p^{-1}\bG^i(\bW^i)'\bG^i(\bW^i))<\lambda_{\max}(p^{-1}\bG^i(\bW^i)'\bG^i(\bW^i))<c_{\max}, \quad a.s. \quad \forall i.\\
$$
(ii) $\max_{k\leq K^i, j\leq p}Eg_k(\bW_j^i)^2<\infty$.
\end{assum}

This assumption is parallel to Assumption \ref{ass2.1.1} (i).  Pervasiveness is trivially satisfied if $\{\bW_j^i\}_{j\leq p}$ are independent and $\bG^i$ is sufficiently smooth.


\begin{assum} \label{ass2.1.3}
(i) $E\gamma_{jk}^i=0$, $\max_{k\leq K^i, j\leq p} |\gamma_{jk}^i| = O_P(\sqrt{\log p})$. \\
(ii) Write $\bgamma_j^i=(\gamma_{j1}^i,...,\gamma_{jK}^i)'$. We assume $\{\bgamma_{j}^i\}_{j\leq p}$ are independent of $\{\bW_j^i\}_{j\leq p}$. \\
(iii) Define
$\nu_p=\max_{i \leq m} \max_{k\leq K^i} p^{-1} \sum_{j\le p}\var(\gamma_{jk}^i) < \infty$. We assume $$\max_{k \le K^i, j\leq p}\sum_{j'\leq p}|E\gamma_{j'k}^i\gamma_{jk}^i |=O(\nu_p).$$
\end{assum}

Condition (i) is parallel to Assumption \ref{ass2.1.1} (ii) whereas Condition (ii) is natural since $\bGamma^i$ can not be explained by $\bW^i$. Condition (iii) imposes cross-sectional weak dependence of $\bgamma_j^i$, which is much weaker than assuming independent and identically distributed $\{\bgamma_j^i\}_{j\leq p}$. This condition is mild as main serial dependency has been taken care of by $g_k(\cdot)$'s.


\section{The ALPHA Framework} \label{sec3}

We introduce the ALPHA framework for heterogeneity adjustment. Methodologically,  for each sub-dataset we aim to estimate the heterogeneity component and subtract it from the raw data.  Theoretically, we aim to obtain the explicit rates of convergence for both the corrected homogeneous signal and its sample covariance matrix. Those rates will be useful when aggregating the homogeneous residuals from multiple sources.

This section covers details for heterogeneity adjustments under both regimes that $\bG^i(\cdot)=0$ and $\bG^i(\cdot)\neq 0$:  they correspond to estimating $\bU^i$ by either PCA or Projected-PCA.
From now on, we drop the superscript $i$ whenever there is no confusion as we focus on the $i^{th}$ data source.  We will use the notation $\widehat{\bF}$ if $\bF$ is estimated by PCA and $\widetilde{\bF}$ if estimated by PPCA.  This convention applies to other related quantities such as $\widehat{\bU}$ and $\widetilde{\bU}$, the heterogeneity-adjusted estimator.  In addition, we use notations such as $\check{\bF}$ and $\check{\bU}$ to denote the final estimators, which are $\widehat{\bF}$ and $\widehat{\bU}$ if PCA is used, and
$\widetilde{\bF}$ and $\widetilde{\bU}$ if PPCA is used.


Estimators for latent factors under regimes 1 and 2 satisfy $n^{-1} \check{\bF}'\check{\bF} = \bI$, which corresponds to normalization in Assumption~\ref{ass2.2} (i). By the principle of least squares, the residual estimator of $\bU$ then admits the form
\begin{equation} \label{eq3.1}
\check\bU = \bX\Big(\bI - \frac{1}{n} \check{\bF}\check{\bF}'\Big).
\end{equation}
It possesses the following properties.

\begin{thm} \label{th3.1}
For any $K$ by $K$ matrix $\bH$ such that $\|\bH\| = O_P(1)$, if $\log p = O(n)$,
$$
\check \bU - \bU = -\frac{1}{n} \bU\bF\bF' + \bPi \,,
$$
where $\|\bPi\|_{\max} = O_P(\sqrt{\log n}/n \cdot (\|\bF'(\check{\bF} - \bF\bH)\|_{\max}\|\bLambda\|_{\max} + \|\bU(\check{\bF} - \bF\bH)\|_{\max}) + \|\check{\bF} - \bF\bH\|_{\max}\|\bLambda\|_{\max} + \sqrt{\log n} \cdot \|\bH\bH' - \bI\|_{\max} \|\bLambda\|_{\max} )$; and furthermore
$$
\check\bU \check\bU'  - \bU \bU' = -\frac{1}{n} \bU\bF\bF'\bU' + \bDelta \,,
$$
where $\|\bDelta\|_{\max} = O_P(\|\bU(\check{\bF} - \bF\bH)\|_{\max} \|\bLambda\|_{\max} + \|\bU(\check{\bF} - \bF\bH)\|_{\max}^2 + \|\bF'(\check{\bF} - \bF\bH)\|_{\max}\|\bLambda\|_{\max}^2 + n\|\bH\bH' - \bI\|_{\max}\|\bLambda\|_{\max}^2)$.
\end{thm}

The above theorem states that the error of estimating $\bU$ by $\check{\bU}$ (or estimating $\bU \bU'$ by $\check{\bU} \check{\bU}'$) is decomposed into two parts. The first part is inevitable even when the factor matrix $\bF$ in (\ref{eq3.1}) is known in advance. The second part is caused by the uncertainty from estimating  $\bF$. Since the true $\bF$ is identifiable up to an orthonormal transformation $\bH$, we need to carefully choose $\bH$  to bound the error $\bPi$ (or $\bDelta$).
We will provide explicit rates of convergence for those terms in the following two subsections.

\subsection{Estimating factors by PCA} \label{sec3.1}


In regime 1, we directly use PCA to adjust data heterogeneity. PCA estimates $\bF$ by $\widehat{\bF}$ where the $k^{th}$ column of $\widehat{\bF}/\sqrt{n}$ is the eigenvector of $(pn)^{-1} {\bX}' \bX$ corresponding to the $k^{th}$ largest eigenvalue.
By the definition of $\widehat{\bF}$, we have $(np)^{-1} \bX'\bX \widehat\bF = \widehat\bF \bK$, where $\bK$ is a $K$ by $K$ diagonal matrix with top $K$ eigenvalues of $(np)^{-1} \bX'\bX$ in descending order as diagonal elements.
Define a $K$ by $K$ matrix $\bH$ as in \cite{FanLiaMin13}:
$$
\bH = \frac{1}{np} \bLambda' \bLambda \bF'\widehat\bF \bK^{-1}\,.
$$
It has been shown that $\|\bK\|$, $\|\bK^{-1}\|$ and $\|\bH\|$, $\|\bH^{-1}\|$ are all $O_P(1)$. The following theorem provides all the rates of convergences that are needed for downstream analysis.

\begin{thm} \label{th3.2} Under Assumptions \ref{ass2.2} and \ref{ass2.1.1}, we have $\|\bLambda\|_{\max} = O_P(\sqrt{\log p})$ and \\
(i) $\|\hF - \bF\bH\|_F = O_P(\sqrt{n/p} + 1/\sqrt{n})$ and $\|\hF - \bF\bH\|_{\max} = O_P(\sqrt{\log n/p} + \sqrt{\log n}/n)$; \\
(ii) $\|\bF'(\hF - \bF\bH)\|_{\max} = O_P(1+\sqrt{n/p})$; \\
(iii)  $\|\bU(\hF - \bF\bH)\|_{\max} = O_P((1+n/p)\sqrt{\log p} + n\|\bSigma\|_1/p)$; \\
(iv) $\|\bH\bH' - \bI\|_{\max} = O_P(1/n + 1/p)$.
\end{thm}

Combining the above results with Theorem \ref{th3.1}, we have
$$
\widehat{\bU}  - \bU  = -\frac{1}{n} {\bU}{\bF}{\bF}' + \bPi\,,
$$
where $\|\bPi\|_{\max} = O_P(\sqrt{\log n \log p}(1/\sqrt{p} + 1/n) + \sqrt{\log n} \|\bSigma\|_1/p)$ and additionally
$$
\widehat{\bU} \widehat{\bU}'  - \bU {\bU}' = -\frac{1}{n} {\bU}{\bF}{\bF}'{\bU}' + \bDelta\,,
$$
where $\|\bDelta\|_{\max} = O_P((1+n/p)\log p + n^2 \|\bSigma\|_1^2/p^2)$. 

\subsection{Estimating factors by Projected-PCA}

In regime 2, we would like to incorporate the external covariates using the Projected-PCA method proposed by \cite{FLW14}. We now explain this method.

To reduce the curse of dimensionality of $g_k(\bW_j)$, we assume it takes an additive form:
\begin{equation} \label{eq2.3}
g_k(\bW_j) = \sum_{l = 1}^d g_{kl}(W_{jl}).
\end{equation}
To model the unknown function $g_{kl}(\cdot)$,  we adopt a sieve based idea which approximates  $g_{kl}(\cdot)$ by a linear combination of basis functions (e.g., B-spline, Fourier series, polynomial series, wavelets).  Let  $\{\phi_1(x),\phi_2(x),\cdots\}$ be a set of basis functions. 
Then for each $l\leq d$,
\begin{equation}\label{eq2.4}
g_{kl}(W_{jl})=\sum_{\nu=1}^J b_{\nu, kl}\phi_\nu (W_{jl})+ R_{kl}(W_{jl}),\quad k\leq K, j\leq p, l\leq d.
\end{equation}
Here $\{b_{\nu,kl}\}_{\nu=1}^J$ are the sieve coefficients of the $l^{th}$ additive component of $g_k(\bW_j)$, corresponding to the $k^{th}$ factor loading;  $R_{kl}$ is the remainder function representing the approximation error; $J$ denotes the number of sieve bases which may  grow slowly as $p$ diverges. The basic assumption for sieve approximation is that   $\sup_{x}|R_{kl}(x)|\rightarrow0$ as $J\rightarrow\infty$. To facilitate notation, we take the same basis functions in (\ref{eq2.4}) for all $k$ and $l$ though they can be different.

Define, for each $k \leq K$ and for each $j\leq p$,
\begin{eqnarray*}
& {\bf b}_k'=( b_{1, k1} , \cdots,  b_{J, k1} , \cdots,  b_{1,kd} , \cdots,  b_{J,kd}) \in \mathbb{R}^{Jd}, \\
& {\phi}(\bW_j)'=( \phi_1(W_{j1}), \cdots,  \phi_J(W_{j1}), \cdots,  \phi_1(W_{jd}) , \cdots,  \phi_J(W_{jd})) \in \mathbb{R}^{Jd}.
\end{eqnarray*}
Then, we can write
$$
g_k(\bW_j)= \phi(\bW_j)'\bb_k+\sum_{l=1}^d R_{kl}(W_{jl}).
$$
Let $\bB=(\bb_1,\cdots,\bb_{K})$ be a $(Jd)\times K$ matrix of sieve coefficients, $\Phi(\bW)=(\phi(\bW_1), \cdots, \phi(\bW_p))'$ be a $p\times (Jd)$ matrix of basis functions, and $\bR(\bW)$ be a $p\times K$ matrix with the $(j,k)^{th}$ element $\sum_{l=1}^d R_{kl}(W_{jl})$. Then the matrix form (\ref{eq2.2}) can be written as
\begin{equation} \label{eq2.5}
\bX = \Phi(\bW)\bB {\bF}' +\bR(\bW) {\bF}' + \bGamma {\bF}' +\bU,
\end{equation}
recalling that we drop the data source index $i$.
Thus the residual term contains three parts: the sieve approximation error $\bR(\bW)\bF'$, unexplained loading $\bGamma\bF'$ and true signal $\bU$.

The idea of Projected-PCA is simple: since the factor loadings are a function of the covariates in \eqref{eq2.5} and $\bU$ and $\bGamma$ are independent of $\bW$, if we project (smooth) the observed data onto the space of $\bW$, the effect of $\bU$ and $\bGamma$ will be significantly reduced and the problem becomes nearly a noiseless one, recalling that the approximation error $\bR(\bW)$ is small.

Define $\bP$ as the projection operator onto the space spanned by the basis functions of $\bW$:
\begin{equation}\label{eq2.7}
\bP=\Phi(\bW)(\Phi(\bW)'\Phi(\bW))^{-1}\Phi(\bW)'.
\end{equation}
Then, by (\ref{eq2.5}), $\bP \bX \approx \bP \Phi(\bW) \bB \bF' \approx \bG(\bW) \bF'$. Thus, $\bF$ can be estimated from the `noiseless data' $\bP \bX$, using the traditional PCA.  Let the columns of $\widetilde\bF/\sqrt{n}$ be the eigenvectors corresponding to the top $K$ eigenvalues of the $n \times n$ matrix $\bX'\bP\bX$, which is the sample covariance matrix of the projected data $\bP \bX$.  Then, $\widetilde{\bF}$ is the PPCA estimator of $\bF$.  It differs from the conventional PCA in that we use smoothed or projected data $\bP \bX$.

By the definition of $\widetilde{\bF}$, we have $(np)^{-1} \bX'\bP\bX \widetilde\bF = \widetilde\bF \bK$ where $\bK$ is a $K \times K$ diagonal matrix with the first $K$ largest eigenvalues of $(np)^{-1} \bX'\bP\bX$ in descending order as its diagonal elements.
Define the $K$ by $K$ matrix $\bH$ as in \cite{FLW14}:
$$
\bH=\frac{1}{np}\bB'\Phi(\bW)'\Phi(\bW)\bB\bF'\widetilde\bF\bK^{-1} \,.
$$
It has been shown that $\|\bK\|$, $\|\bK^{-1}\|$ and $\|\bH\|$, $\|\bH^{-1}\|$ are all $O_P(1)$.  Here we remind that though $\bH$ and $\bK$ are different from those in regime 1, they play essentially the same roles (thus with same notations).

As in \cite{FLW14}, we need the following conditions for the basis functions and accuracy of the sieve approximation.

\begin{assum}[Basis functions]\label{ass2.3}
(i) There are $d_{\min}$ and $d_{\max}>0$ so that almost surely,
 $$
d_{\min}<\lambda_{\min}(p^{-1}\Phi(\bW)'\Phi(\bW))<\lambda_{\max}(p^{-1}\Phi(\bW)'\Phi(\bW))<d_{\max}.
$$
(ii) $\max_{\nu \leq J, j\leq p, l\leq d}E\phi_\nu(W_{jl})^2<\infty.$
\end{assum}

\begin{assum}[Accuracy of sieve approximation]\label{ass2.4} For each $l\leq d, k\leq K$,\\
  (i) The sieve coefficients $\{b_{\nu,jl}\}_{\nu=1}^J$ satisfy: $\exists \, \kappa \geq 4$, as $J\rightarrow\infty$,
      $$
  \sup_{x\in\mathcal{X}_l}|g_{kl}(x)-\sum_{\nu=1}^Jb_{\nu,kl}\phi_\nu(x)|^2=O(J^{-\kappa}),
      $$
    where $\mathcal{X}_l$ is the support of the $l^{th}$ element of $\bW_j$, and $J$ is the sieve dimension. \\
    (ii)  $\max_{\nu, k, l} |b_{\nu, kl}| <\infty$.
\end{assum}

Condition (i) in Assumption \ref{ass2.4}  is  satisfied by  most commonly used basis. For example,   when $\{\phi_\nu\}$ is polynomial basis or B-splines,  it is implied by the condition that smooth curve $g_{kl}(\cdot)$ belongs to a H\"{o}lder class $\mathcal{G}$, defined by $\mathcal{G}=\{g: |g^{(r)}(s)-g^{(r)}(t)|\leq L|s-t|^{\alpha} \}$ for some $L > 0$, with $\kappa=2(r+\alpha) \geq 4$ \citep{Lor86, Che07}. Another example is step function $g_{kl}(\cdot)$ with finite many distinct values, which can be expressed exactly as the linear combination of disjoint indicator functions so that $\kappa$ can be arbitrarily large. 

With the above conditions, the following theorem provides all the rates we need, recalling the definition of $\nu_p$ in Assumption \ref{ass2.1.3} (iii).

\begin{thm} \label{th3.3} Choose $J=(p\min\{n,p, \nu_p^{-1}\})^{1/\kappa}$ and assume $J^2\phi_{\max}^2 \log (nJ)=O(p)$ where $\phi_{\max} = \max_{\nu \le J} \sup_{x \in \mathcal X} \phi_{\nu}(x)$. Under Assumptions \ref{ass2.2}, \ref{ass2.1.2}, \ref{ass2.1.3}, \ref{ass2.3} and \ref{ass2.4}, 
we have $\|\bLambda\|_{\max} = O_P(J \phi_{\max} + \sqrt{\log p})$ and \\
(i) $\|\tF - \bF\bH\|_F = O_P(\sqrt{n/p})$ and $\|\tF - \bF\bH\|_{\max} = O_P(\sqrt{\log n/p})$; \\
(ii) $\|\bF'(\tF - \bF\bH)\|_{\max} = O_P(\sqrt{n/p}+ n/p+n\sqrt{\nu_p/p})$; \\
(iii)  $\|\bU(\tF - \bF\bH)\|_{\max} = O_P(\sqrt{n \log p/p} + nJ\phi_{\max}\|\bSigma\|_1/p)$; \\
(iv) $\|\bH\bH' - \bI\|_{\max} = O_P(1/p+1/\sqrt{pn} + \sqrt{\nu_p/p})$.
\end{thm}

Combining the above theorem with Theorem \ref{th3.1}, we obtain
$$
\widetilde{\bU}  - \bU  = -\frac{1}{n} {\bU}{\bF}{\bF}' + \bPi\,,
$$
where $\|\bPi\|_{\max} = O_P(\sqrt{\log n/p} (J\phi_{\max} + \sqrt{ \log p}) + J\phi_{\max} \|\bSigma\|_1\sqrt{\log n} /p)$ and
$$
\widetilde{\bU} \widetilde{\bU}'  - \bU {\bU}' = -\frac{1}{n} {\bU}{\bF}{\bF}'{\bU}' + \bDelta\,,
$$
where $\|\bDelta\|_{\max} = O_P(n\sqrt{\nu_p/p}(J^2\phi_{\max}^2 + \log p)+nJ\phi_{\max}\|\bSigma\|_1 (J\phi_{\max} + \sqrt{\log p})/p + n^2J^2\phi_{\max}^2\|\bSigma\|_1^2/p^2)$ if there exists C s.t. $\nu_p > C/n$. We choose to keep $\|\bSigma\|_1$ terms here although it makes a long presentation of the rate.

\subsection{Specification test}


In this section, we give an adaptive rule to decide whether the covariates $\bW$ are informative enough to use PPCA or just PCA.  We test the hypothesis that $H_0: \bG(\bW) = 0$ using the test statistic \citep{FLW14}
$$
S = \frac{1}{p} \tr(\bXi \widehat{\bLambda}' \bP \widehat{\bLambda}) \;\; \text{where} \;\; \bXi = \Big(\frac1p \widehat{\bLambda}' \widehat{\bLambda} \Big)^{-1} \,.
$$
Here, we use PCA estimator $\widehat{\bLambda}$ as PPCA is not applicable under $H_0$. If $\bLambda$ has nothing to do with $\bW$, then $\bP \widehat \bLambda \approx 0$ and $S$ should be quite small after projection. Conversely, if  $\bG(\bW)\neq 0$, $S$ will be large, hence we reject the null. We showed the following theorem, whose proof is omitted.

\begin{thm}\label{th5.1}
Under all the assumptions discussed above, if $\{\bW_j, \bgamma_j\}_{j\le p}$ are independent and identically distributed, as $p,n_i,J \to \infty$, we have under $H_0$ for the $i^{th}$ subgroup,
$$
\frac{p S - J d K}{\sqrt{2JdK}} \overset{d} \to  N(0,1)\,.
$$
\end{thm}

Based on this result, we can decide whether or not to reject the null hypothesis, namely to use PPCA or PCA.  When the test is applied to all $m$ data sources, it becomes a multiple testing problem. The thresholding can be chosen by using various false discovery rate control methods such as Benjamini-Hochberg method \citep{BenHoc95}. If a hypothesis is rejected, we identify the subgroup as regime 2 and use Projected-PCA to obtain $\widetilde{\bU}$; otherwise, we identify the subgroup as regime 1 and apply regular PCA to get $\widehat{\bU}$.

\subsection{Estimating number of factors}

We now address the problem of estimating the number of factors for two different regimes.  Extensive literature has made contributions to this problem in regime 1, i.e. the regular factor model \citep{BN02, HL, AH13, LamYao12}.  \cite{AH13} and \cite{LamYao12} proposed to use ratio of adjacent eigenvalues of ${\bX}'\bX$ to infer the number of factors. 
They showed the estimator $\widehat K = \arg \max_{k \le K_{\max}} \lambda_{k}({\bX}'\bX)/\lambda_{k+1}({\bX}'\bX)$ correctly identifies $K$ with probability tending to $1$, where $K_{\max}$ can be a fixed prior upper bound for the number of factors.

For the geniune semiparametric factor model, in the recent work by \cite{FLW14},
they propose $\widetilde K = \arg \max_{k \le K_{\max}} \lambda_{k}({\bX}' \bP \bX)/\lambda_{k+1}({\bX}' \bP \bX)$. Here $K_{\max}$ is of the same order as $Jd$, say $K_{\max} = Jd/2$. It was shown that $\mathbb P(\widetilde K = K) \to 1$ under assumptions we omit here. When we have genuine and pervasive covariates, $\widetilde K$ typically outperforms $\widehat K$. More details can be found in \cite{FLW14}.


\subsection{Summary of ALPHA} \label{sec3.5}
We now summarize the final procedure and convergence rates. We first divide $m$ subgroups based on whether the collected covariates have influence on the loadings. Let
\begin{align*}
\mathcal M_1 & = \{i \le m \; |\; \bG = 0 \}\,, \qquad
\mathcal M_2  = \{i \le m \; |\; \bG \ne 0 \}\,.
\end{align*}

ALPHA consists of the following three steps.
\begin{itemize}[leftmargin=1.5cm]
\item[Step 1:] ({\bf Preprocessing}) For data source $i$, estimate $K^i$ by $\widehat K^i$; test $H_0: \bG^i(\bX^i) = 0$ by $S^i$ and use FDR control to construct two groups.  For rejected groups, refine $\widehat K^i$ by $\widetilde K^i$.

\item[Step 2:] ({\bf Adjustment}) Apply Projected-PCA to estimate $\bLambda^i {\bF^i}'$ if the test is rejected, otherwise use PCA to remove the heterogeneity, resulting in adjusted data $\check{\bU}$, which is either $\widehat\bU^i$ or $\widetilde\bU^i$.

\item[Step 3:] ({\bf Aggregation}) Combine adjusted data $\{\check \bU^i\}_{i=1}^m$ to conduct further statistical analysis. For example, estimate sample covariance $\bSigma$  by $\widehat \bSigma = (N - \sum_i {\widehat {K_i}})^{-1} \sum_{i=1}^m {\check{\bU^i}} {\check {\bU^i}}'$
    where 
    $N = \sum_i n_i$ is the aggregated sample size; or estimate sparse precision matrix $\bOmega$  by existing graphical model methods say CLIME \citep{CLL11}.
\end{itemize}

We summarize the ALPHA procedure in Algorithm \ref{algo1} given in the appendix.

We also summarize the convergence of $\widehat\bU^i$ and $\widetilde\bU^i$ here. To ease presentation, we consider a typical regime in practice: $n_i < Cp$, $\sum_{i \le m} K^i < CN$ for some  constant $C$. Also we focus on the situation of sufficiently smooth curves $\kappa = \infty$ so that $J$ diverges very slowly (say with rate $O(\sqrt{\log p})$) and constant $\phi_{\max}, \nu_p$.
Based on discussions of the previous subsections, for estimation of $\bU$, we have
$$
\check{\bU}^i - \bU^i = -{\bU^i}{\bF^i}{\bF^i}'/n_i +
\begin{cases}
       O_P\Big(\sqrt{\log n_i \log p/p} + \sqrt{\log n_i \log p}/n_i \Big) & \text{if} \;\; i \in \mathcal M_1\,, \\
       O_P\Big(\sqrt{\log n_i \log p/p} \Big) & \text{if} \;\; i \in \mathcal M_2\,. \\
   \end{cases}
$$
Therefore,  Projected-PCA  dominates PCA as long as the effective covariates are provided. However, ${\bU^i}{\bF^i}{\bF^i}'/n_i$ dominates all the remaining terms so that $\|\check{\bU}^i - \bU^i\|_{\max} = O_P(\|{\bU^i}{\bF^i}{\bF^i}'/n_i\|_{\max}) = O_P(\sqrt{\log n_i \log p/n_i})$.

In addition, for estimation of $\bU\bU'$, we have
\begin{equation} \label{eq3.6}
\check{\bU}^i \check{\bU}^i{}' - \bU^i {\bU^i}' = -{\bU^i}{\bF^i}{\bF^i}'{\bU^i}'/n_i +
\begin{cases}
       O_P\Big(\log p + \delta \Big) & \text{if} \;\; i \in \mathcal M_1\,, \\
       O_P\Big(n_i\log p \sqrt{\nu_p/p} +\delta \Big) & \text{if} \;\; i \in \mathcal M_2\,, \\
   \end{cases}
\end{equation}
where $\delta = n_i^2 \|\bSigma\|_1^2  \log p/p^2$, depending on $\|\bSigma\|_1$. 
If we consider a general $\bSigma$ so that $\|\bSigma\|_1$ can be as large as $O(\sqrt{p})$, then the rate for $i \in \mathcal M_1$ is simplified to $O_P((1 + n_i^2/p)\log p)$ while the rate for $i \in \mathcal M_2$ is $O_P((n_i/\sqrt{p} + n_i^2/p)\log p)$. This illustrates the advantage of Projected-PCA since its convergence is faster.
If we only consider very sparse covariance matrix so that $\|\bSigma\|_1$ is bounded, we can simply drop the term $\delta$ in both regimes.
Then, regime 1 achieves better rate if $p = O(n_i^2 \nu_p)$, but regime 2 dominates otherwise. 


\section{Conditional Graphical Model} \label{sec4}

We have summarized the order of bias caused by adjusting heterogeneity for each data source in Section \ref{sec3.5}. Now we combine the adjusted data together for further statistical analysis. As an example, we study graph estimation under a Gaussian graphical model.

Assume $\bu_t^i \sim N({\bf 0},\bSigma)$ and consider the class of the precision matrices:
\begin{equation} \label{eq2.6}
\mathcal F(s, R) = \Big\{ \bOmega: \bOmega \succ {\bf 0}, \; \|\bOmega\|_1 \le R, \; \max_{1\le i \le p} \sum_{j = 1}^p \mathbbm{1} (\Omega_{ij} \ne 0) \le s \Big\}.
\end{equation}
To simplify the analysis, we assume $R$ is fixed, but all the analysis can be easily extended to include growing $R$.

To estimate $\bOmega = \bSigma^{-1}$ via CLIME, we first need a covariance estimator as the input. 
We also assume here the number of factors is known, i.e., the exception probability of recovering $K^i$ has been ignored for ease of discussion. Such an estimator is natually given by
\begin{equation} \label{CovEst}
\widehat\bSigma = \frac{1}{N - \sum_{i \le m} K^i} \sum_{i = 1}^m \check{\bU^i} {\check{\bU^i}}'\,.
\end{equation}
Since the number of data sources is large, we focus on the typical case of diverging $N$ and $p$.

\subsection{Covariance estimation}\label{sec4.1}

Denote by $\bSigma_N$ the oracle sample covariance matrix i.e. $\bSigma_N = N^{-1} \sum_{i = 1}^m \bU^i {\bU^i}'$. We consider the difference of our proposed $\widehat\bSigma$ with $\bSigma_N$ in this subsection. The oracle estimator obviously attains the rate $\|\bSigma_N - \bSigma\|_{\max} = O_P(\sqrt{\log p/N})$.

Let ${\bxi}_{k}^i = \bU^i \bar{\bff_k^i}/\sqrt{n_i}$ where $\bar{\bff_k^i}$ is the $k^{th}$ column of $\bF^i$. It is Gaussian distributed with mean zero and variance $\bSigma$. Note that ${\bxi}_k^i$ are iid with respect to $k$ and $i$, using the assumption ${\bF^i}'\bF^i/n_i = \bI$. By the standard concentration bound,
$$
\Big\|\sum_{i \le m} \Big( \frac{1}{n_i}{\bU^i}{\bF^i}{\bF^i}'{\bU^i}' - K^i \bSigma\Big)\Big\|_{\max} = \Big\|\sum_{i \le m} \sum_{k \le K^i} \Big({\bxi}_k^i {{\bxi}_k^i}' - \bSigma \Big) \Big\|_{\max} = O_P\Big(\sqrt{K^{tot} \log p}\Big)\,,
$$
where $K^{tot} = \sum_{i\le m} K^i$. Therefore, by \eqref{eq3.6}, we have
\begin{equation} \label{eq4.1}
\begin{aligned}
\|\widehat\bSigma - \bSigma_N\|_{\max} = & \; \Big\| \frac{N}{N-\sum_{i\le m} K^i} \frac{1}{N} \sum_{i \le m} \Big(\check{\bU^i} \check{\bU^i}{}' - \bU^i {\bU^i}' + K^i \bSigma\Big) \\
& \; + \frac{\sum_{i\in \mathcal M} K^i}{N - \sum_{i\in \mathcal M} K^i} \Big(\frac1N \sum_{i \le m} \bU^i {\bU^i}' - \bSigma\Big) \Big\|_{\max} \\
=&:  O_P(a_{m,N,p})\,,
\end{aligned}
\end{equation}
where $a_{m,N,p}=
\frac{|\mathcal M_1| \log p}{N}  +  \frac{N_2 \log p}{N} \sqrt{\frac{\nu_p}{p}}  + \frac{\sqrt{K^{tot} \log p}}{N}   + \frac{K^{tot}}{N} \sqrt{\frac{\log p}{N}}$ and $N_2 = \sum_{i \in \mathcal M_2}  n_i$.

We now examine the difference of the ALPHA estimator from the oracle estimator for two specific cases.  In the first case, we apply PCA to all data sources, i.e., all $ i \in \mathcal M_1$ and $K^i$ is bounded. We then have $a_{m, N, p} = m \log p/N$. This rate is dominated by the oracle error rate $\sqrt{\log p/N}$ if and only if $m = O(\sqrt{N/\log p})$. This means traditional PCA performs optimally for adjusting heterogeneity as long as the number of subgroups grows slower than the order of $\sqrt{N/\log p}$.

If we apply PPCA to all data sources, i.e., $i \in \mathcal M_2$ and $K^i$ is bounded,  then $a_{m, N, p} = \sqrt{\nu_p/p} \log p+ \sqrt{m \log p}/N$. This rate is of smaller order than  rate $\sqrt{\log p/N}$ if $p/\log p > CN$ for some constant $C > 0$. The advantage of using PPCA is that  when $n_i$ is bound so that  $m \asymp N$, we can still achieve optimal rate of convergence so long as we have a large enough dimensionality at least of the order $N$.


\subsection{Precision matrix  estimation}

In order to obtain an estimator for the sparse precision matrix from $\widehat\bSigma$, we apply the CLIME estimator proposed by \cite{CLL11}. For a given $\widehat\bSigma$, CLIME solves the following optimization problem:
\begin{equation} \label{eq2.10}
\widehat\bOmega = \arg \min_{\bOmega} \|\bOmega\|_{1,1} \quad \text{subject to} \quad \|\widehat\bSigma \bOmega -\bI\|_{\max} \le \lambda,
\end{equation}
where $\|\bOmega\|_{1,1} = \sum_{i,j \le p} |\sigma_{ij}|$ and $\lambda$ is a tuning parameter. Note that (\ref{eq2.10}) can be solved column-wisely by linear programming.
However, CLIME does not necessarily generate a symmetric matrix. We can simply symmetrize it by taking the one with minimal magnitude of $\hat\sigma_{ij}$ and $\hat\sigma_{ji}$. The resulting matrix after symmetrization, still denoted as $\widehat\bOmega$ with a little bit abuse of notation, also attains good rate of convergence. In particular, we consider the sparse precision matrix class $\mathcal F(s, C_0)$ in (\ref{eq2.6}). The following lemma provides guarantee for recovering any sparse matrix $\bOmega \in \mathcal F(s, C_0)$.

\begin{thm}\label{th2.3}
Suppose $\bOmega \in \mathcal F(s, C_0)$ and $\widehat \bSigma$ given by (\ref{CovEst}) attains the rate $\|\widehat\bSigma - \bSigma_N\|_{\max} = O_P(a_{m, N, p})$ in (\ref{eq4.1}) with $\bSigma_N$ denoting oracle sample covariance matrix. Letting $\tau_{m,N,p} = \sqrt{\log p/N} + a_{m, N, p}$ and $\lambda \asymp \tau_{m,N,p}$, we have
$$
\|\widehat\bOmega - \bOmega\|_{\max} = O_p(\tau_{m,N,p}).
$$
Furthermore,
$$
\|\widehat\bOmega - \bOmega\|_1= O_p(s \tau_{m,N,p}) \quad \mbox{and} \quad \|\widehat\bOmega - \bOmega\|_2 = O_p(s \tau_{m,N,p}).
$$
\end{thm}

The proof of the theorem can be found in the appendix. The theorem shows that CLIME has strong theoretical guarantee of convergence under different matrix norms. The rate of convergence has two parts, one corresponding to the minimax optimal rate \citep{Yua10} while the other is due to the error caused by estimating the unknown factors under various situations.  The discussions at the end of Section \ref{sec4.1} suggests that the latter error is often negligible.

\section{Numerical Studies} \label{sec5}

In this section, we first validate the theoretical results derived above through Monte Carlo simulations. Our purpose is to show that after heterogeneity adjustment, our proposed aggregated covariance estimator $\widehat\bSigma$  approximates well the oracle sample covariance $\bSigma_N$, thereby leading to accurate estimation of the true covariance matrix $\bSigma$ and precision matrix $\bOmega$. We also compare the performance of Projected-PCA and regular PCA on heterogeneity adjustments under different asymptotic settings.

In addition, we analyze a real brain image data using the proposed procedure. The dataset to be analyzed is the ADHD-200 data \citep{BMZ10}. It consists of rs-fMRI images of 688 subjects, of whom 491 are healthy and 197 are diagnosed with ADHD. We dropped 16 subjects (13 healthy, 3 diseased) in our analysis since their data contain missing values. Following \cite{PCN11}, we divided the whole brain into 264 regions of interest (ROI, $p=264$), which are regarded as nodes in our graphical model. Each brain is scanned for multiple times with sample sizes ranging from 76 to 261 ($76 \le n_i \le 261$). In each scan, we acquire the blood-oxygen-level dependent (BOLD) signal within each ROI. Here the heterogeneity among subjects arises from the difference in age, gender, handedness and IQ.

\subsection{Preliminary analysis}\label{sec6.1}

To analyze the data, the first question is what external covariates $\bW_j^i$ are for each of the 264 regions. 
Ideally, we hope these covariates have pervasive power on explaining the batch effect, while bearing no association with the graph structure of $\bu_t$. 
For the current data, we can construct such covariates from physical locations of the regions, since the level of batch effect is non-uniform over different locations of the brain when scanned in fMRI machines, and furthermore it has been widely acknowledged in biological study that spatial adjacency does not necessarily imply brain functional connectivity. 

Here we simply split the 264 regions into 10 clusters ($J=10$) by the hierarchy clustering (Ward's minimum variance method) of their physical locations and use the categorical cluster indices as the covariates of the nodes. Note that the healthy and ADHD group share the same physical locations. 
The clustering result is shown in Figure \ref{fig:0} and the spatial locations of the 264 regions are shown in Figure \ref{fig:4} in 10 different colors. Black (middle), green (left) and blue (right) represent roughly the region of frontal lobe; gray (middle), pink (left) and magenta (right) occupy the region of parietal lobe; red (left) and orange (right) are in the area of occipital lobe; finally yellow (left) and navy (right) provide information about temporal lobe.

Other possible values for $J$ could also be considered, but we do not want $J$ too large to overfit the smooth loading functions. Note that here since the covariate $W$ is one-dimensional ($d=1$) and discrete, the sieve basis functions are just indicator functions $\mathbbm{1}(w-0.5 \le W < w+0.5)$ for $w=1, \dots, 10$. We use the same covariates for all subjects.

\begin{figure}[ht]
	\centering
	\includegraphics[scale=.5]{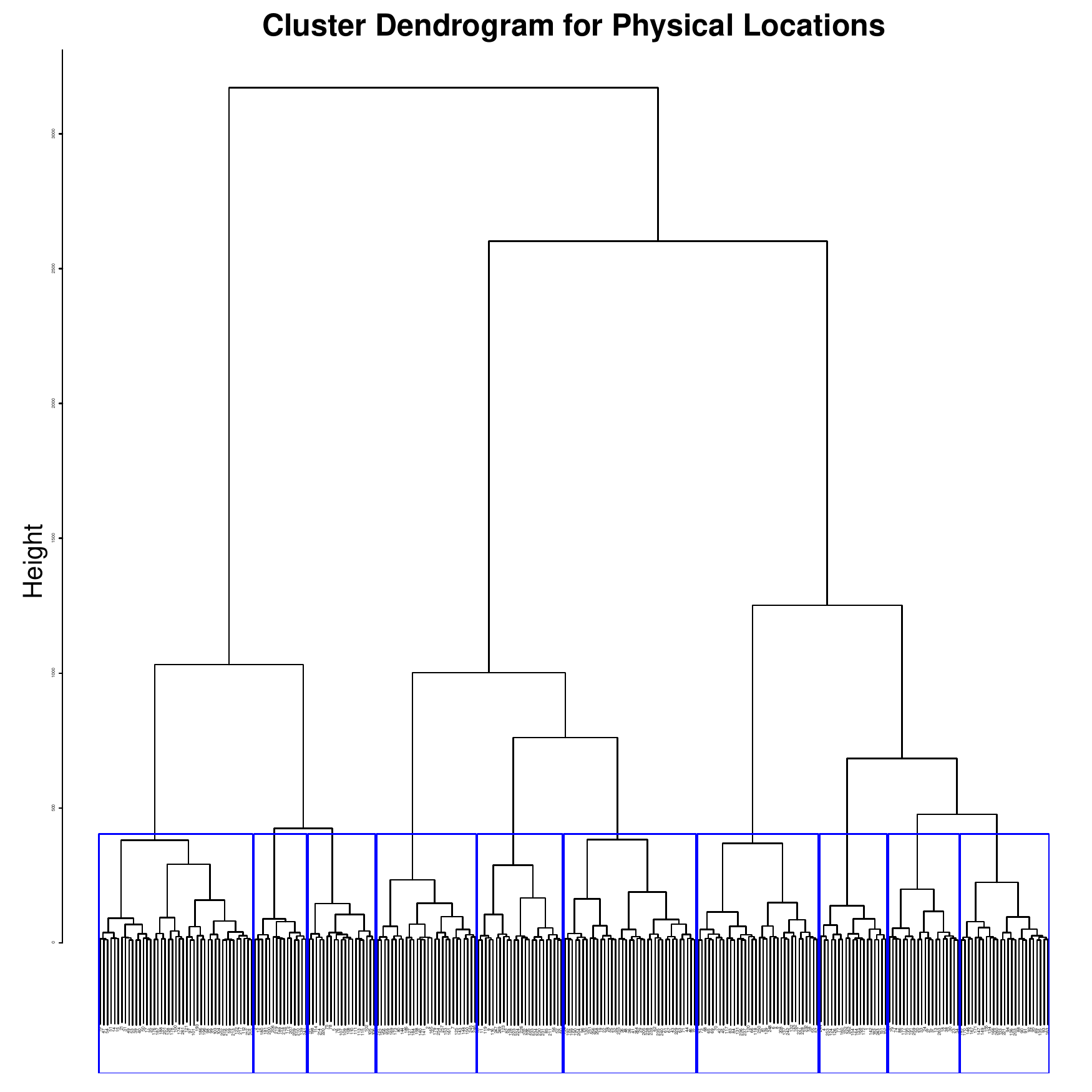}	
	\caption{Cluster Dendrogram for $J = 10$ for physical locations.} \label{fig:0}
\end{figure}

The next question is whether the selected covariates can explain the loadings well. We implemented the specification tests described in Section \ref{sec5} and find out the p-values for each subject. Most of the p-values are rather small ($82.4\%$ subjects in the healthy group and $79.0\%$ subjects in the patient group have p-values smaller than $10^{-3}$). We chose to control the FDR by Benjamini-Hochberg method \citep{BenHoc95} below the level of $1\%$. We discovered 425 healthy samples ($91.4\%$) and 129 diseased samples ($90.2\%$) rejecting the null, meaning that the selected covariates have significant explanatory powers on factor loadings of most subjects. We identified them as samples in the class $\mathcal M_2$ and used Projected-PCA to estimate the heterogeneity effect. For those whose null hypotheses were not rejected, we classified them as individuals in the class $\mathcal M_1$ and regular PCA was applied.

Based on which class each subject falls into, we employ the corresponding method to estimate the number of factors. We used $K_{\max} = 5$. The estimated number of factors for the two groups are summarized in Table \ref{tab:numFactors}.

\begin{table}[ht]
  \centering
    \caption{The distributions of the estimated number of factors for healthy and ADHD groups}  \label{tab:numFactors}
  \begin{tabular}{cccccc}
    \hline
    $\widehat K^i$ & 1  & 2 & 3 & 4 & 5\\
    \hline
    Healthy & 227 & 126 & 59 & 31 & 22\\
    ADHD & 67 & 34 & 23 &  12 & 7  \\
    \hline
  \end{tabular}
  \end{table}

\subsection{Synthetic datasets}\label{sec6.2}

In this simulation study, for stability, we use the first 15 subjects in the healthy group to calibrate the simulation models. The testing results reveal that the external covariates $\bW$ are informative for each of these 15 subjects. We specify four asymptotic settings for our simulation studies:
\begin{enumerate}
	\item $m=500$, $n_i=10$ for $i=1,..,m$, $p=100, 200, ..., 600$ and $\bG(\bW)\neq 0$;

	\item $m=100, 200, ..., 1000$, $n_i=10$ for $i=1,...,m$, $p=264$ and $\bG(\bW)\neq 0$;
	
	\item $m=100$, $n_i=10, 20, ..., 100$ for $i=1,...,m$, $p=264$ and $\bG(\bW)\neq 0$;
	
	\item $m=20, 40, ..., 200$, $n_i=20, 40, ..., 200$ for $i=1, ..., m$, $p=264$ and $\bG(\bW)=0$.
\end{enumerate}

Here the last setting represents regime 1 with $\bG(\bW) = 0$ where we should expect PCA to work well when the number of subjects is of order of square root of the total sample size, that is $m \asymp \sqrt{N}$.
The first three settings represent regime 2 with informative covariates $\bG(\bW)\neq 0$; they present asymptotics with growing $p$, $m$ and $n_i$ respectively.

\subsubsection{Model calibration and data generation}

We calibrate (estimate) the covariance matix $\bSigma$ of $\bu_t$, which is a $264$ by $264$ matrix, by our proposed method to the data in the healthy group. Plugging it as input in CLIME solver delivers a sparse precision matrix $\bOmega$, which will be taken as truth in the simulation. Note that due to the regularization in CLIME, $\bOmega^{-1}$ is not the same as $\bSigma$.  To obtain the covariance matrix used in setting 1, we also calibrate, using the same method, a sub-model that involves only the first 100 regions.  We then copy this $100\times100$ matrix multiple times to form a $p\times p$ block diagonal matrix and used it for simulations in setting 1.
We describe how we calibrate these `true models' and generate data from the models as follows.
\begin{enumerate}

	\item {\bf (External covariates)} For each $j \le p$, generate the external covariate $W$ i.i.d. from the multinomial distribution with $\mathbb P (W_j=s)=w_s, s\le 10$ where $\{w_s\}_{s=1}^{10}$ are calibrated with the hierarchy clustering results of the real data (Figure~\ref{fig:0}).
	
	\item {\bf (Calibration)} For the first 15 healthy subjects, obtain estimators for $\bF$, $\bB$ and $\bGamma$ by PPCA, resulting in $\widetilde{\bF}$, $\widetilde\bB = n^{-1} (\Phi(\bW)'\Phi(\bW))^{-1} \Phi(\bW)'\bX\widetilde{\bF}$ and $\widetilde\bGamma = n^{-1} (\bI - \bP) \bX \widetilde{\bF}$ according to \cite{FLW14}.  Use the rows of the estimated factors to fit a stationary VAR model $\bff_t=\bA \bff_{t-1}+\bepsilon_t$, where $\bepsilon_t \sim N(0, \bSigma_{\epsilon})$, and obtain the estimators $\widetilde\bA$ and $\widetilde\bSigma_{\epsilon}$.
		
   \item {\bf (Simulation)} For each subject $i \leq m$, pick one of the 15 calibrated models and their associated parameters from above at random and do the following.

	\begin{enumerate}
		\item Generate $\gamma_{jk}^i$ i.i.d. from $N(0, \tilde\sigma_{\gamma}^2)$ where $\tilde\sigma_{\gamma}^2$ is  the variance of all entries of $\widetilde\bGamma$. For the first three settings, compute the `true' loading matrix $\bLambda^i=\Phi(\bW)\widetilde\bB+\bGamma^i$. For the last setting, set $\bLambda^i=\bGamma^i$ since $\bG(\bW) = 0$.

		\item Generate factors $\bff ^i_t$ from the VAR model $\bff ^i_t=\widetilde\bA \bff ^i_{t-1}+\bepsilon_t$ with $\bepsilon_t \sim N(0, \widetilde\bSigma_{\epsilon})$, where the parameters $\widetilde\bA$ and $\widetilde\bSigma_{\epsilon}$ are taken from the fitted values in step 2.
		
		\item Finally, generate the observed data $\bX^i=\bLambda^i\bF^{i'}+\bU^i$, where each column of $\bU^i$ is randomly sampled from $N({\bf 0}, \bOmega^{-1})$,
where $\bOmega$ has been calibrated by the CLIME solver as described at beginning of the section.
		
	\end{enumerate}
	
\end{enumerate}

\subsubsection{Estimation of $\bSigma$}

\begin{figure}[ht]
	\centering
	\begin{tabular}{cc}
		$\qquad\|\widehat\bSigma-\bSigma_N\|_{\max} \qquad\qquad \|\widehat\bSigma-\bSigma\|_{\max}$ & $\qquad\|\widehat\bSigma-\bSigma_N\|_{\max} \qquad\qquad \|\widehat\bSigma-\bSigma\|_{\max}$ \\
		\includegraphics[scale=.45]{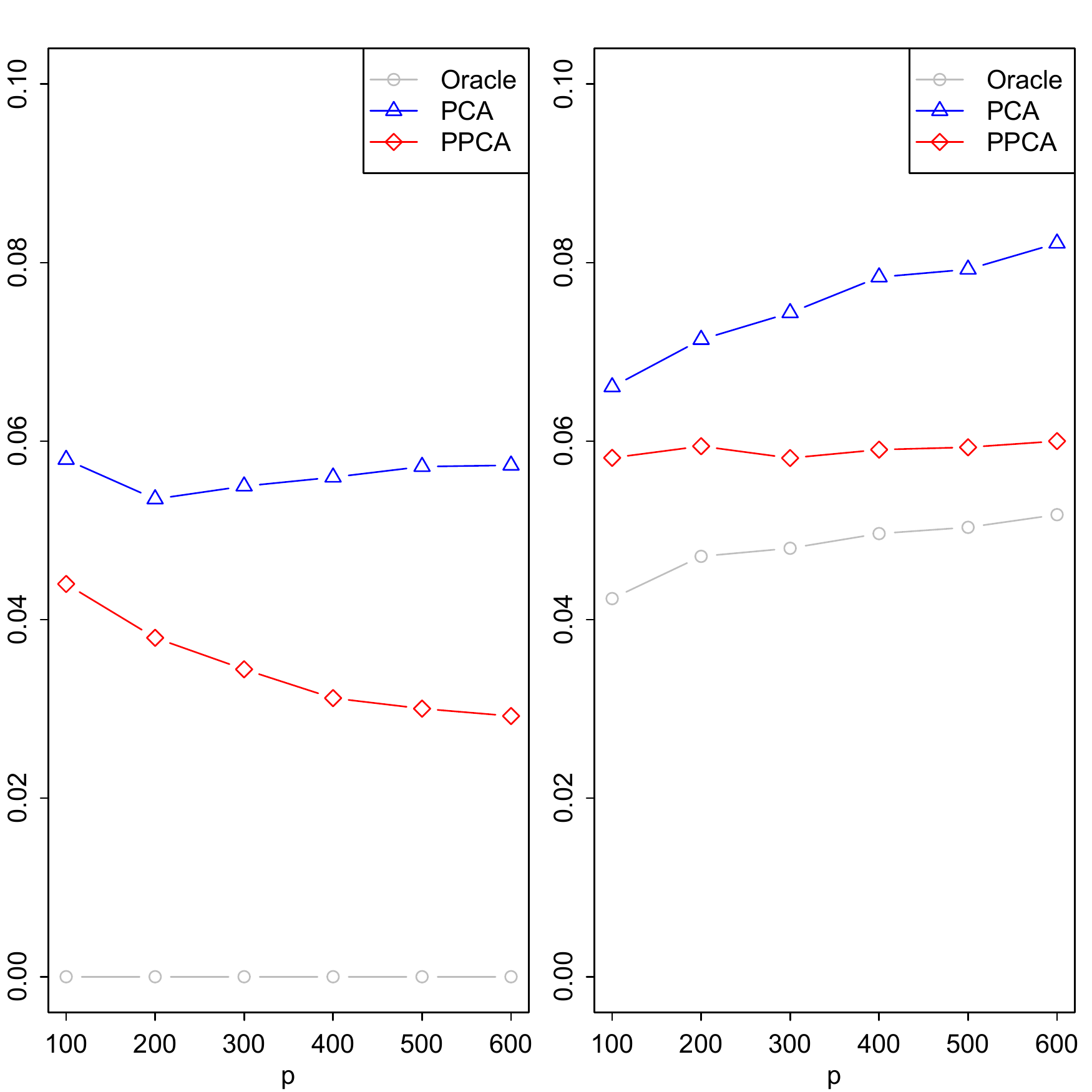} & \includegraphics[scale=.45]{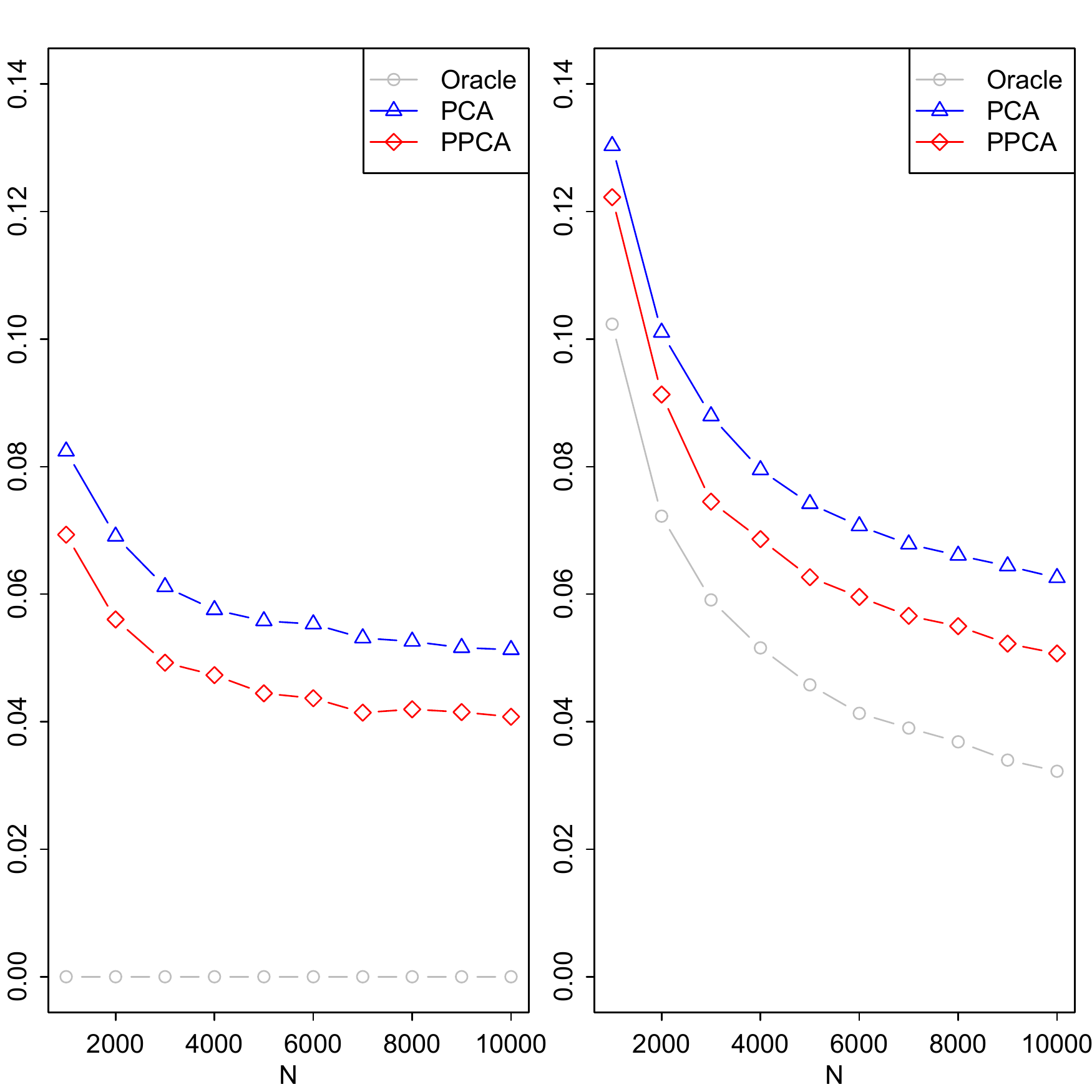}\\
		Case 1 & Case 2 \vspace{.2in} \\
		$\qquad\|\widehat\bSigma-\bSigma_N\|_{\max} \qquad\qquad \|\widehat\bSigma-\bSigma\|_{\max}$ & $\qquad\|\widehat\bSigma-\bSigma_N\|_{\max} \qquad\qquad \|\widehat\bSigma-\bSigma\|_{\max}$ \\
		\includegraphics[scale=.45]{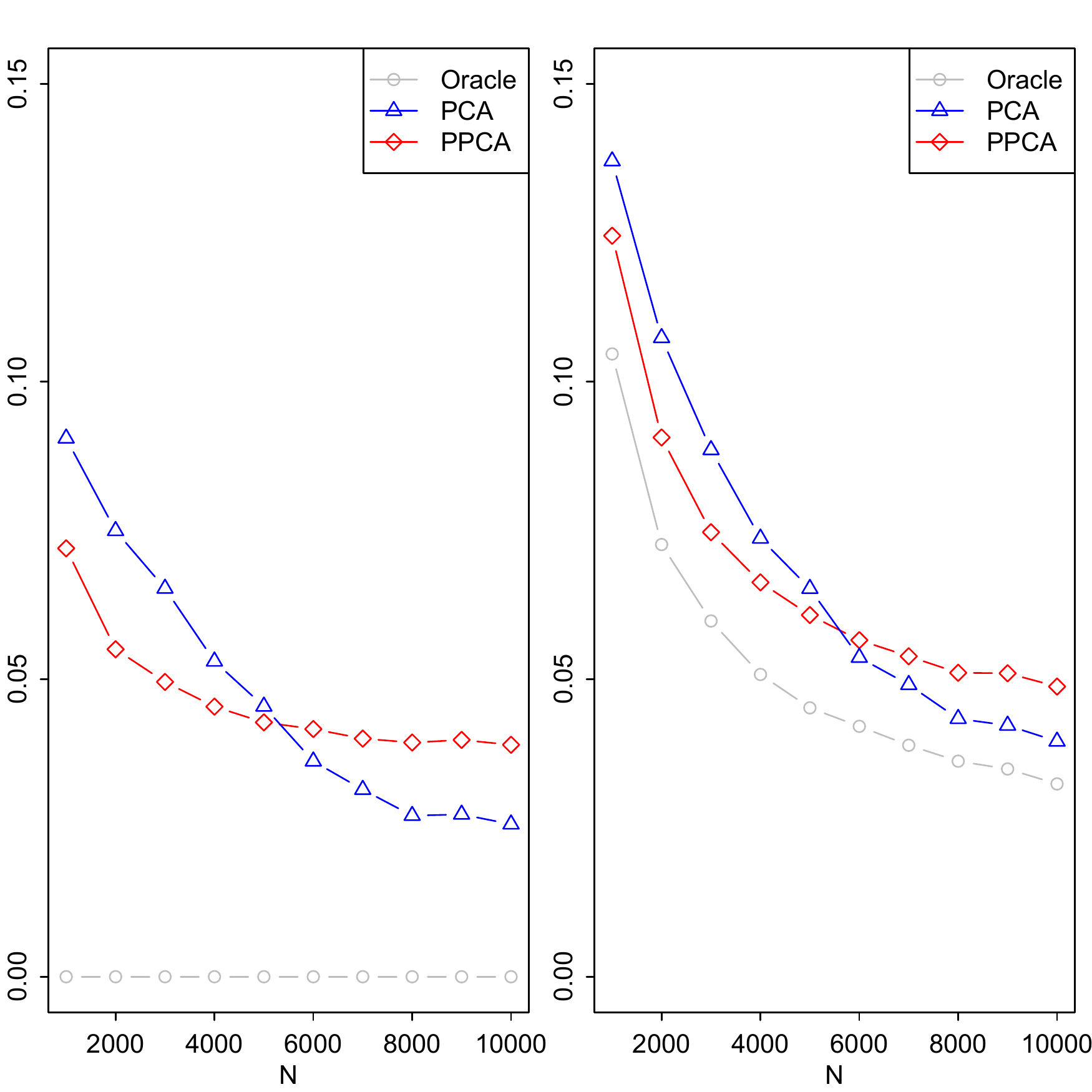} & \includegraphics[scale=.45]{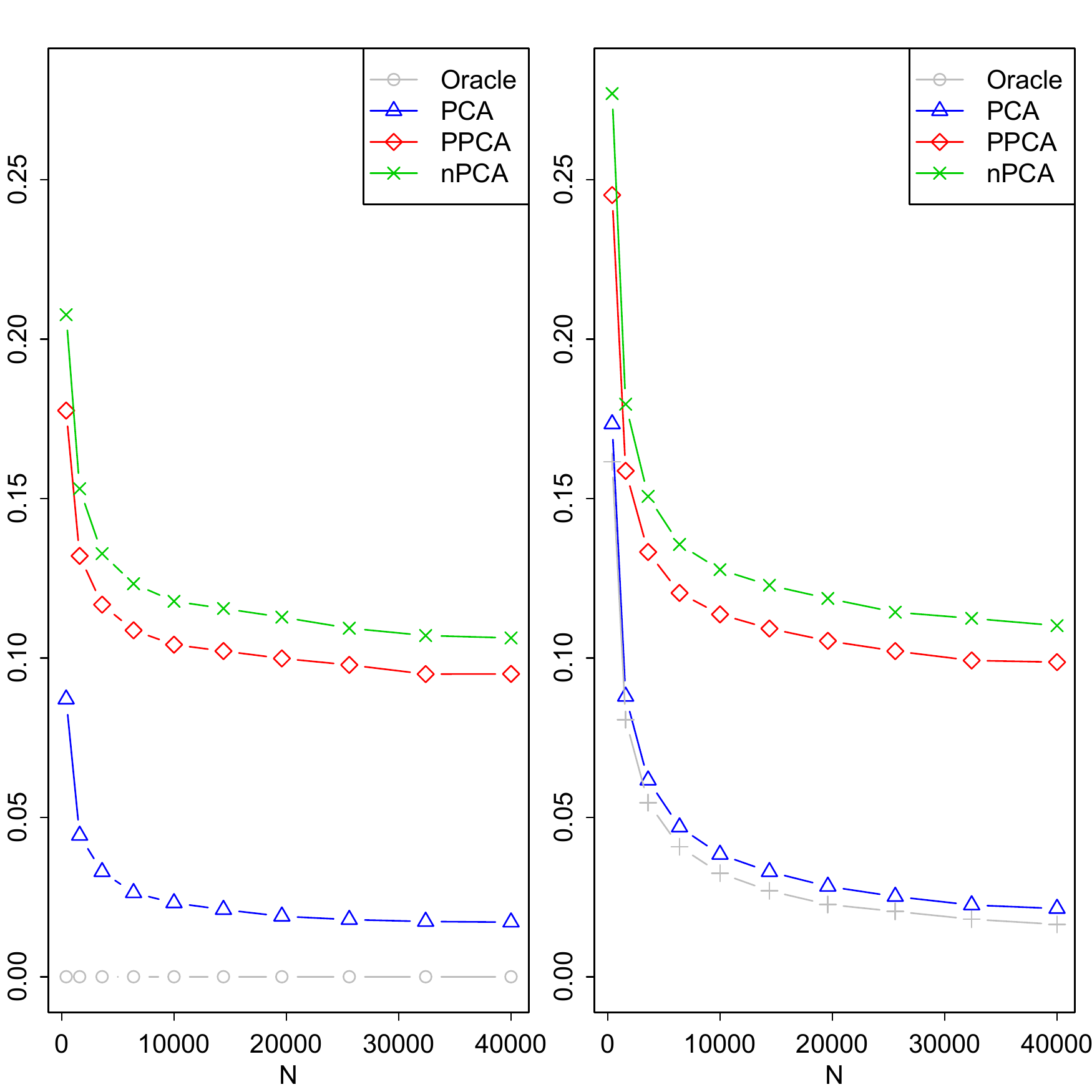}\\
		Case 3 & Case 4
	\end{tabular}
	\caption{Estimation of $\bSigma$ by PCA, PPCA and the oracle sample covariance matrix for 4 different settings. Case 1:  $m$ and $n_i$ are fixed while the dimension $p$ increases; case 2: $n_i$ and $p$ are fixed while $m$ increases; case 3: $m$ and $p$ are fixed while $n_i$ increases; case 4:  $p$ is fixed, and both $m$ and $n_i$ increase and conditions for PPCA are violated. } 	
	\label{fig:1}
\end{figure}

In this subsection, we investigate the errors of estimating covariance of $\bu_t$ in max-norm after applying Projected-PCA or regular PCA for heterogeneity adjustment. We also compare them with the estimation errors if we naively pool all the data together without any heterogeneity adjustment, but the estimation errors for the first 3 cases are too large to fit in the graph. Denote the oracle sample covariance of $\bu_t$ by $\bSigma_N$ as before. The estimation errors under all four settings are presented in Figure \ref{fig:1}, which are based on 100 simulations.

In Case 1, $m$ and $n_i$ are fixed while the dimension $p$ increases. For this setting, $n_i$ is small and this highlights more the advantages of Projected-PCA over regular PCA. From the left panel, we observe that increase of dimensionality improves the performance of Projected-PCA. This is consistent with the rate we derived in theories.
In Case 2, $n_i$ and $p$ are fixed while $m$ increases. Both Projected-PCA and regular PCA benefit from increasing number of subjects. However, since $n_i$ is small, again Projected-PCA outperforms regular PCA. In Case 3, $m$ and $p$ are fixed while $n_i$ increases. Again both methods achieve better estimation as $n_i$ increases, but more importantly, regular PCA outperforms Projected-PCA when $n_i$ is large enough. This is again consistent with our theories.  As illustrated by Section $4.1$, when $m$ is fixed, PCA attains the convergence rate $\|\widehat\bSigma-\bSigma\|_{\max}=O_P(\sqrt{\log p/N})$, while Projected-PCA only achieves $\|\widehat\bSigma-\bSigma\|_{\max}=O_P(\log p/\sqrt{p})$, which is worse than PCA when $p/\log p = o(N)$. In Case 4, $p$ is fixed, and both $m$ and $n_i$ increase. Note that the covariates have no explanation power at all, i.e., Condition \ref{ass2.1.2} about pervasiveness does not hold so that PPCA is not applicable. As expected, adjusting by PCA behaves much better than by Projected-PCA, which can sometimes be as bad as `nPCA', corresponding to no heterogeneity adjustment.  This is not unexpected as we utilized a noisy external covariates.

\subsubsection{Estimation of $\bOmega$}

\begin{figure}[ht]
	\centering
	\begin{tabular}{cc}
		$\qquad\|\widehat\bOmega-\bOmega\|_{\max} \qquad\qquad \|\widehat\bOmega-\bOmega\|_{1}$ & $\qquad\|\widehat\bOmega-\bOmega\|_{\max} \qquad\qquad \|\widehat\bOmega-\bOmega\|_{1}$ \\
		\includegraphics[scale=.45]{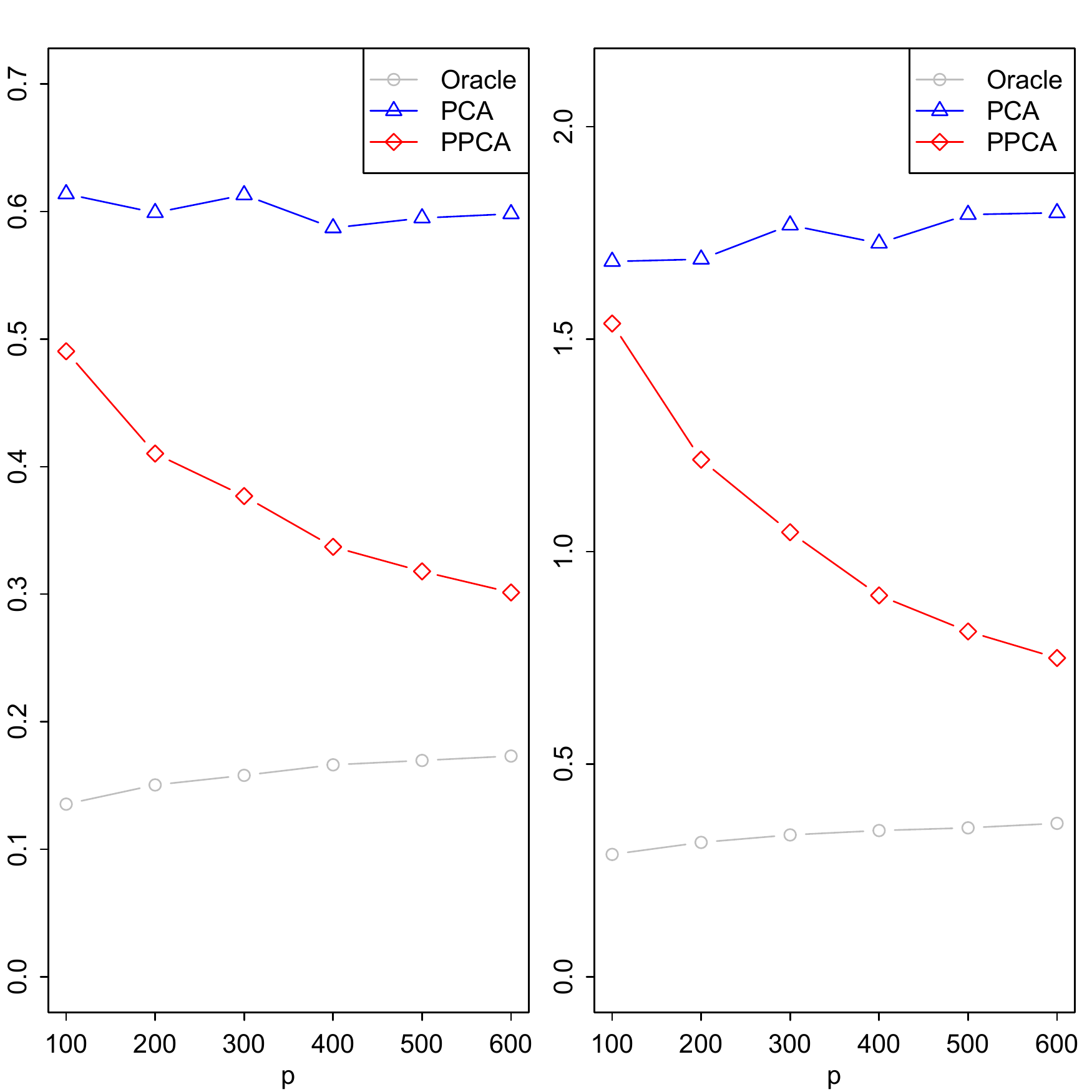} & \includegraphics[scale=.45]{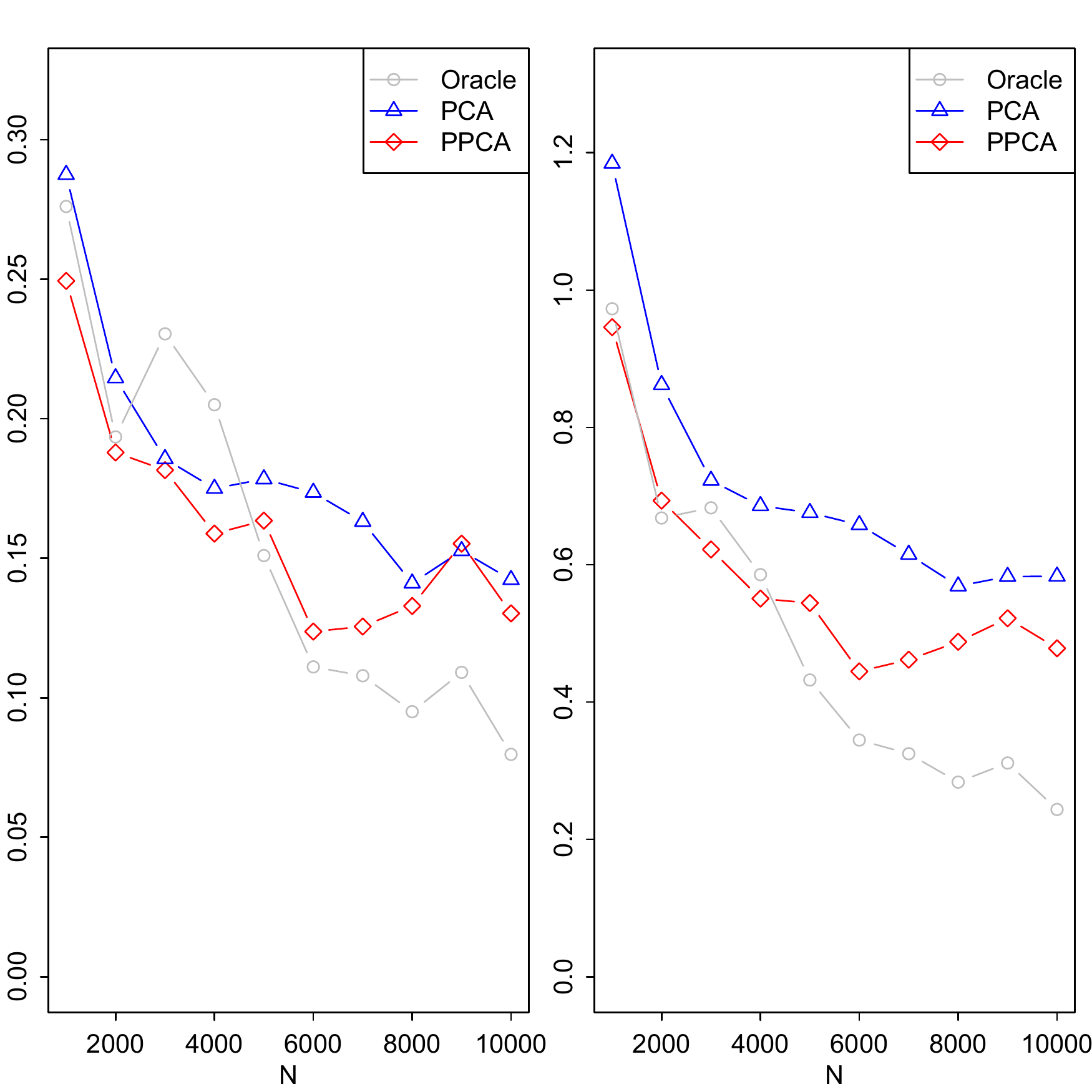}\\
		Case 1 & Case 2 \vspace{.2in} \\
		$\qquad\|\widehat\bOmega-\bOmega\|_{\max} \qquad\qquad \|\widehat\bOmega-\bOmega\|_{1}$ & $\qquad\|\widehat\bOmega-\bOmega\|_{\max} \qquad\qquad \|\widehat\bOmega-\bOmega\|_{1}$ \\
		\includegraphics[scale=.45]{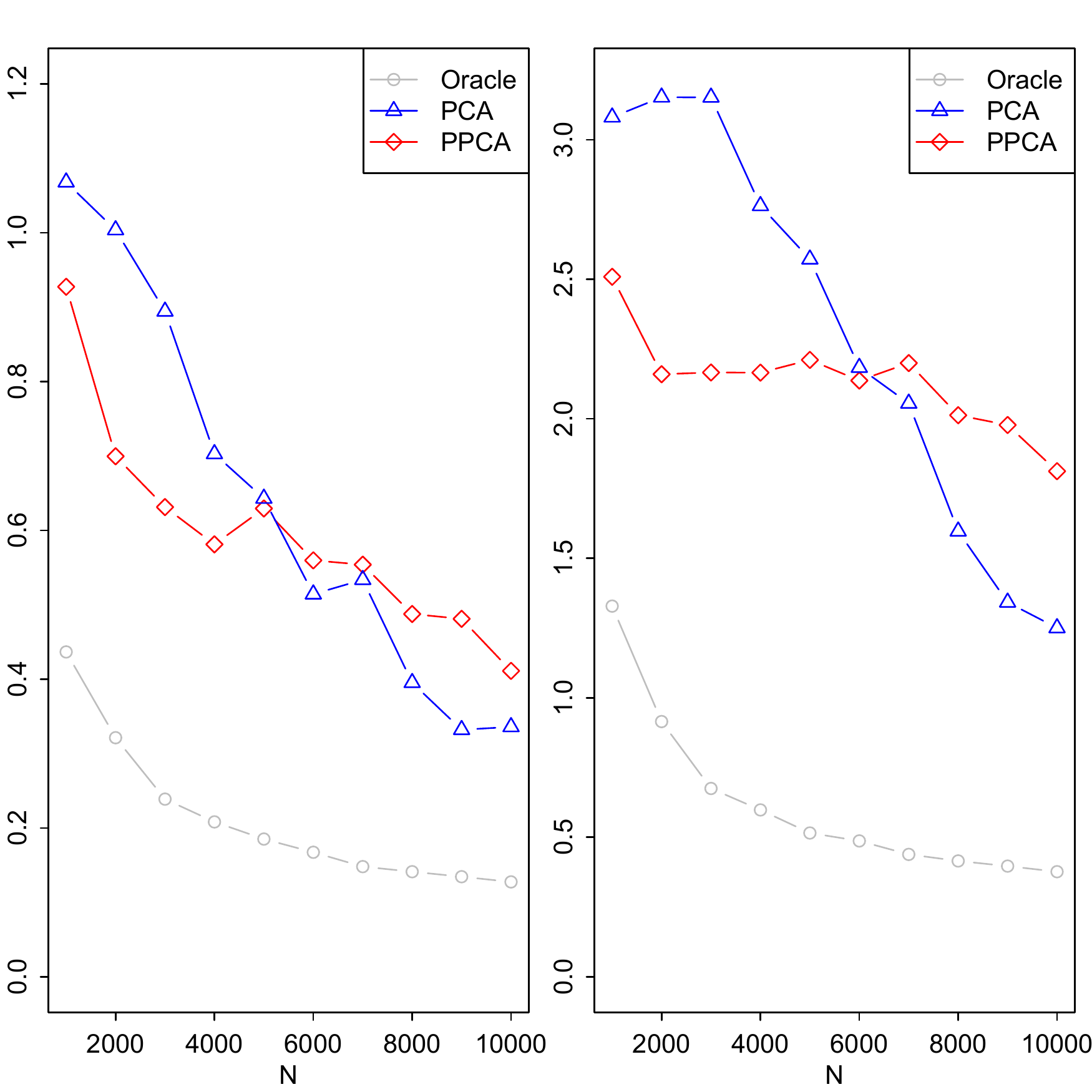} & \includegraphics[scale=.45]{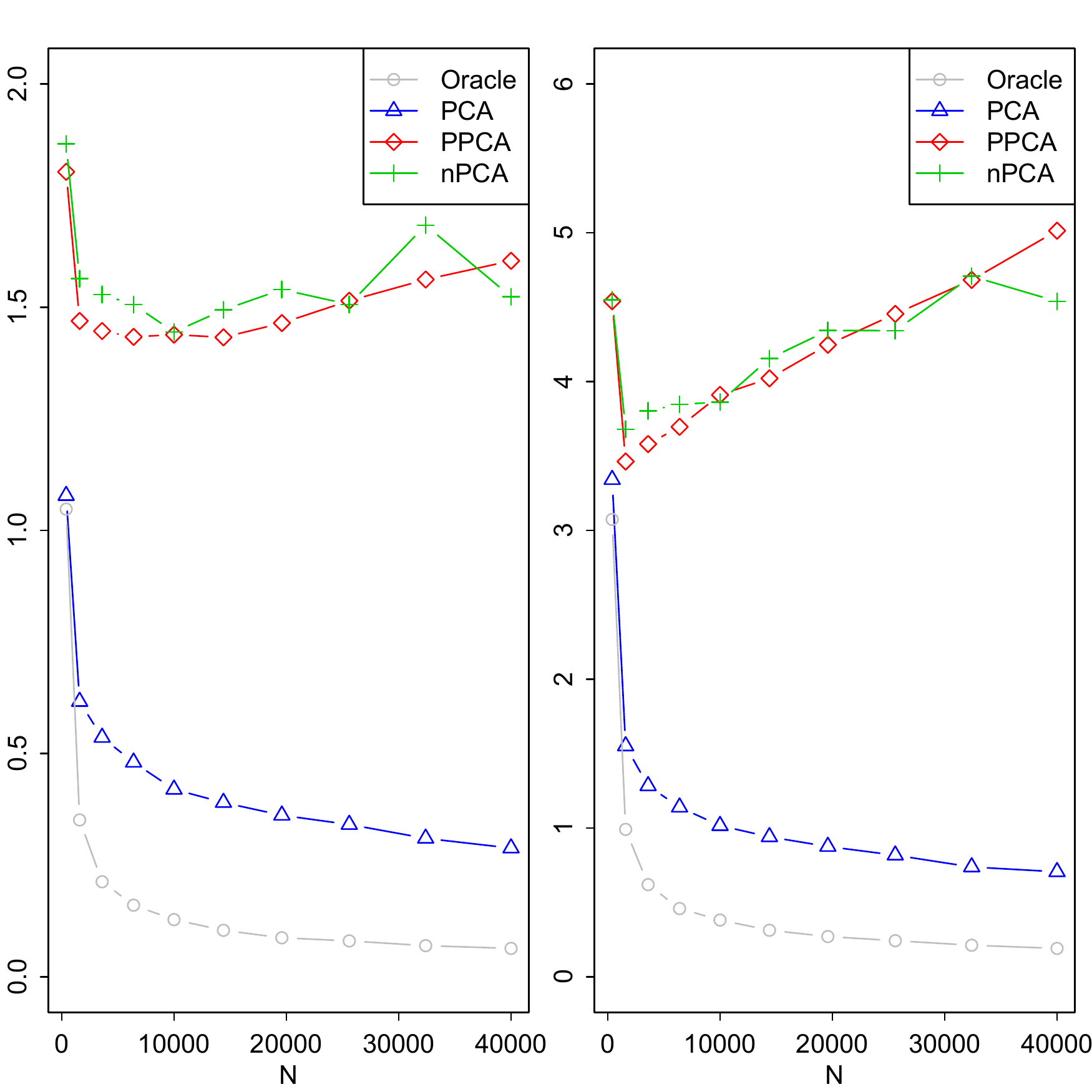}\\
		Case 3 & Case 4
	\end{tabular}
	\caption{Estimation of $\bOmega$.   Presented are the estimation errors in max-norm and in $L_1$-norm for 4 different settings.  The same captions in Figure~\ref{fig:1} apply.}
	\label{fig:2}
\end{figure}

In this subsection, we focus on estimation error of the precision matrix of $\bu_t$. We plug $\widehat\bSigma$, obtained from data after adjusting for heterogeneity, into CLIME to get an estimator $\widehat\bOmega$ of $\bOmega$. In Figure \ref{fig:2}, $\|\widehat\bOmega-\bOmega\|_{\max}$ and $\|\widehat\bOmega-\bOmega\|_1$ are depicted under the same four asymptotic settings as before. From the plots we see $\|\widehat\bOmega-\bOmega\|_{\max}$ and $\|\widehat\bOmega-\bOmega\|_{1}$ share similar behavior with $\|\widehat\bSigma-\bSigma\|_{\max}$ in all the four settings. In the first three cases, if we do not adjust data heterogeneity, $\|\widehat\bOmega-\bOmega\|_{\max}$ and $\|\widehat\bOmega-\bOmega\|_1$ will be too large to be fitted in the current plots.

\begin{figure}[ht]
	\centering
	\begin{tabular}{cc}
		\includegraphics[trim=0cm 0 0 0, scale=.45, clip]{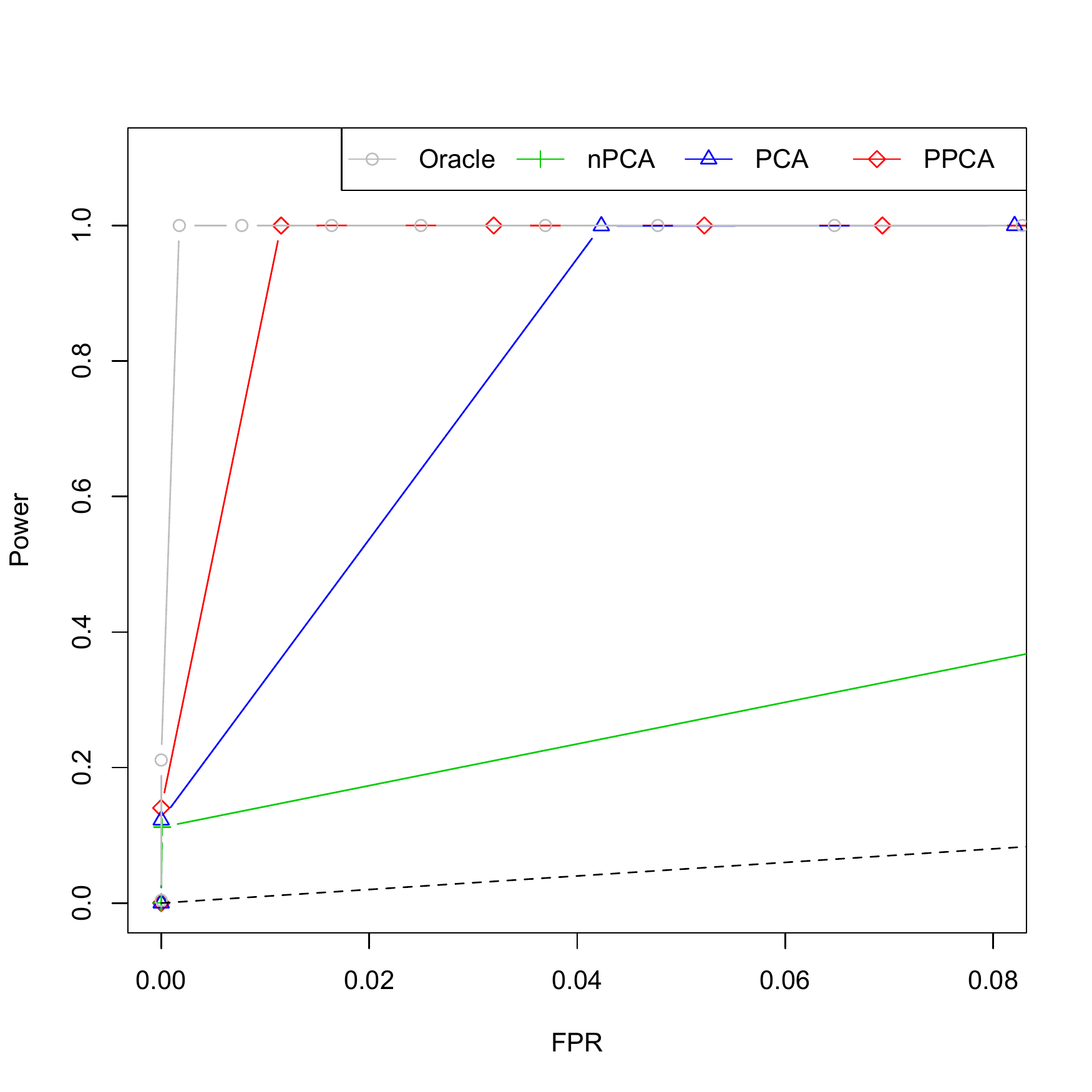} & \includegraphics[trim=0cm 0 0 0, scale=.45, clip]{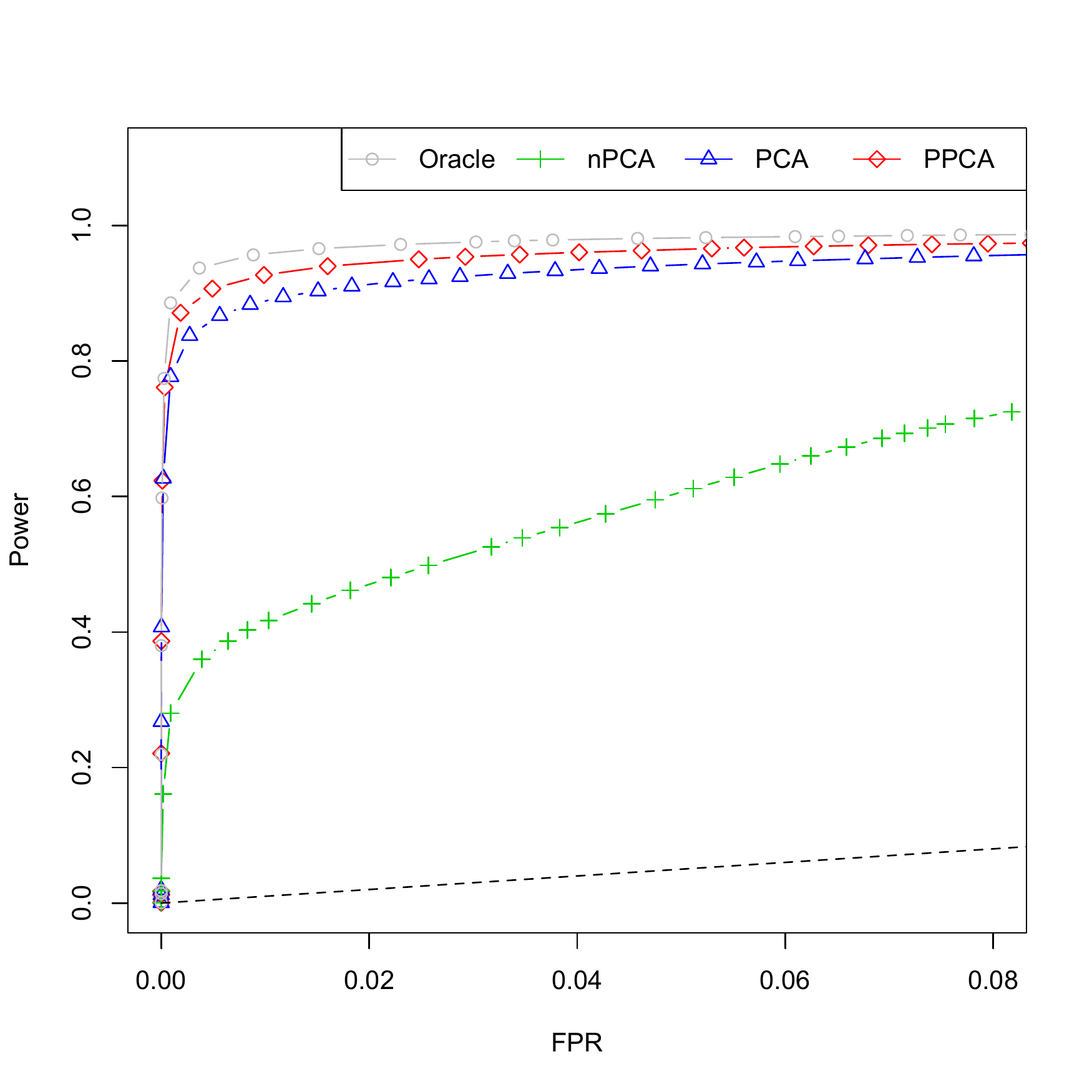}\\
		Case 1: $m=500$, $n_i=10$, $p=100$ & Case 2: $m=400$, $n_i=10$, $p=264$ \vspace{.2in} \\
		\includegraphics[scale=.45]{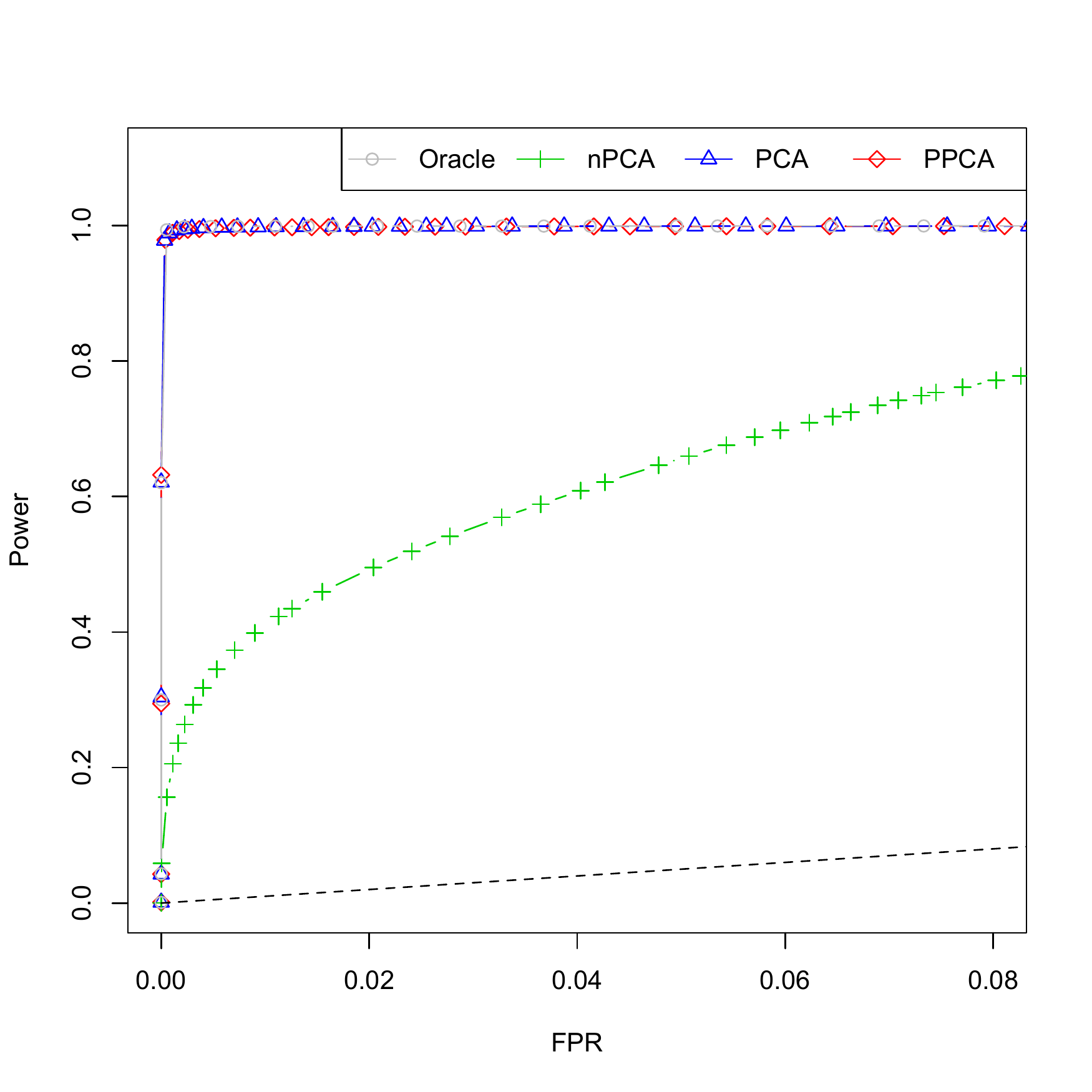} & \includegraphics[scale=.45]{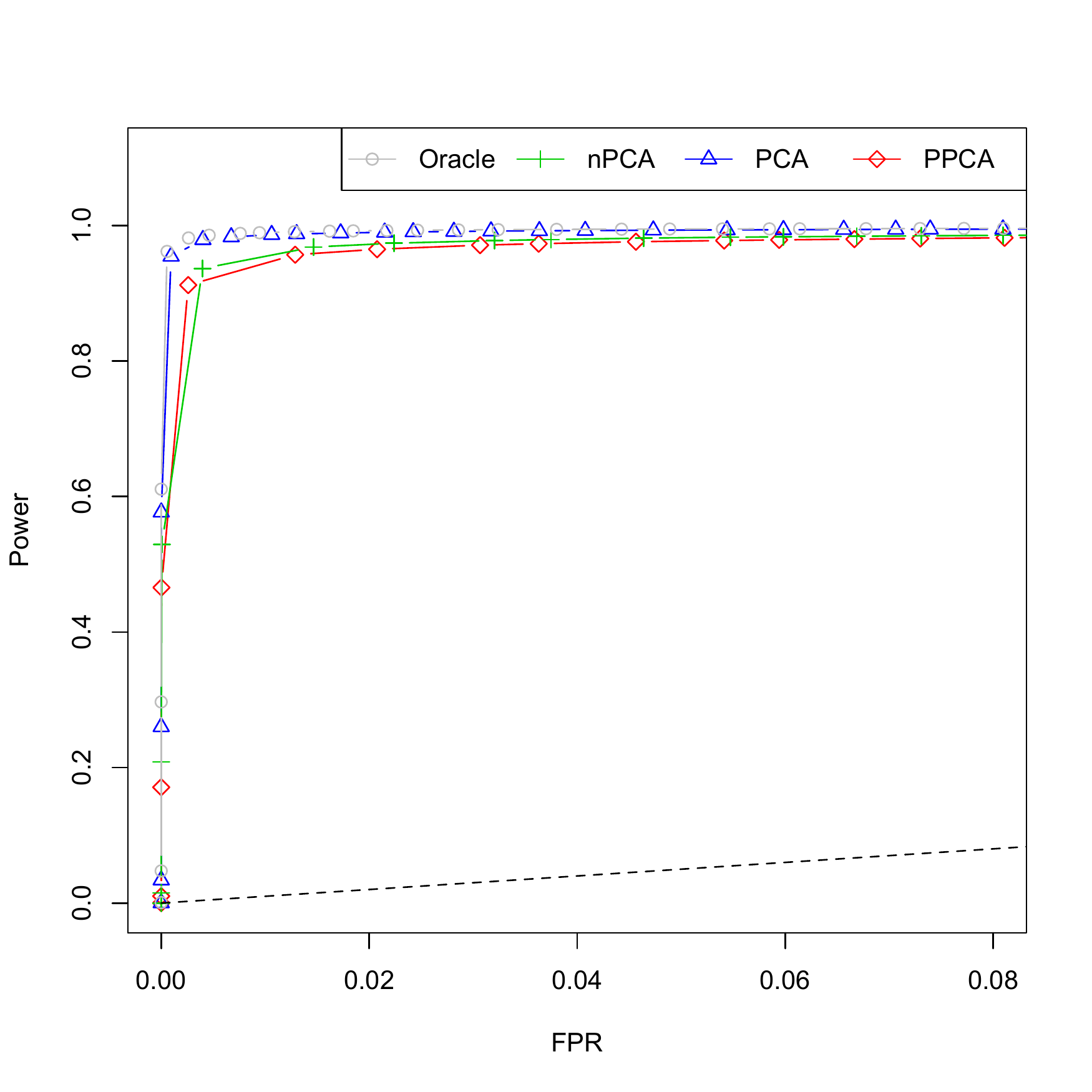}\\
		Case 3: $m=100$, $n_i=100$, $p=264$ & Case 4: $m=60$, $n_i=60$, $p=264$
	\end{tabular}
	\caption{ROC curves for sparsity recovery of $\bOmega$ for 4 different settings.  The captions in Figure~\ref{fig:1} apply.}
		\label{fig:3}
\end{figure}

We also present the ROC curves of our proposed methods in Figure \ref{fig:3}, which is of interest to readers concerned with sparsity pattern recovery. The black dashed line is the 45 degree line connecting $(0,0)$ and $(1,1)$, representing performance of the random guess. It is obvious from those plots that heterogeneity adjustment very much improves the sparsity recovery of the precision matrix $\bOmega$. When the sample size of each subject is small, genuine pervasive covariates increase the power of Projected-PCA method while on the other hand if the sample size is relatively large, PCA is sufficiently good in recovering graph structures. Also notice that in all cases, the naive method with no heterogeneity adjustment can still achieve a certain amount of power, but we can improve the performance dramatically by correcting the batch effects.

\subsection{Brain image network data}\label{sec6.3}

We report the estimated graphs for both the healthy group and the ADHD patient group with batch effects removed using three methods:
(1) PPCA using physical locations of ROI as covariates;
(2) PCA without using any covariates;
(3) no-PCA, which
ignores heterogeneity and naively pool the data from all subjects together. We took various sparsity levels of the networks from $1\%$ to $5\%$ (corresponding to the same set of $\lambda$'s for two groups) and selected the common edges, which are stable with respect to tuning, to be depicted.

\begin{figure}[ht]
	\centering
	\begin{tabular}{ccc}
		\includegraphics[trim=0cm 0 0 0, scale=.2, clip]{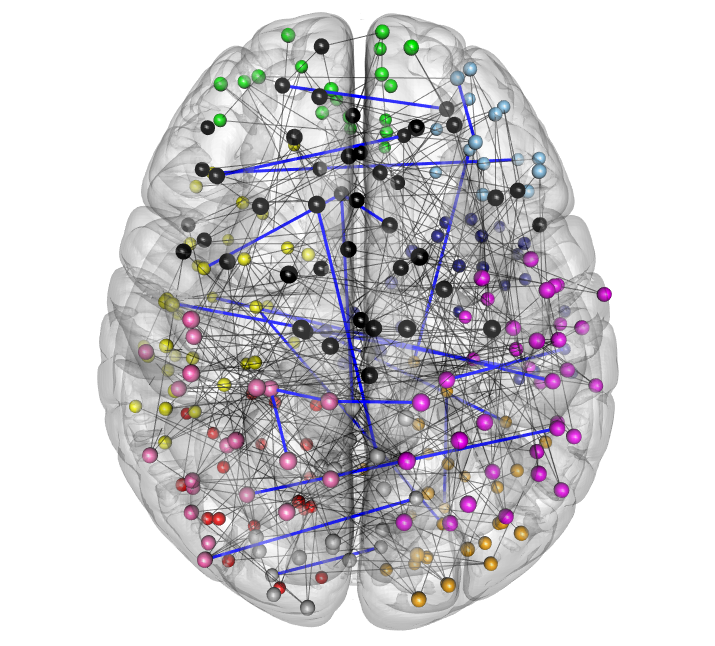} & \includegraphics[trim=0cm 0 0 0, scale=.2, clip]{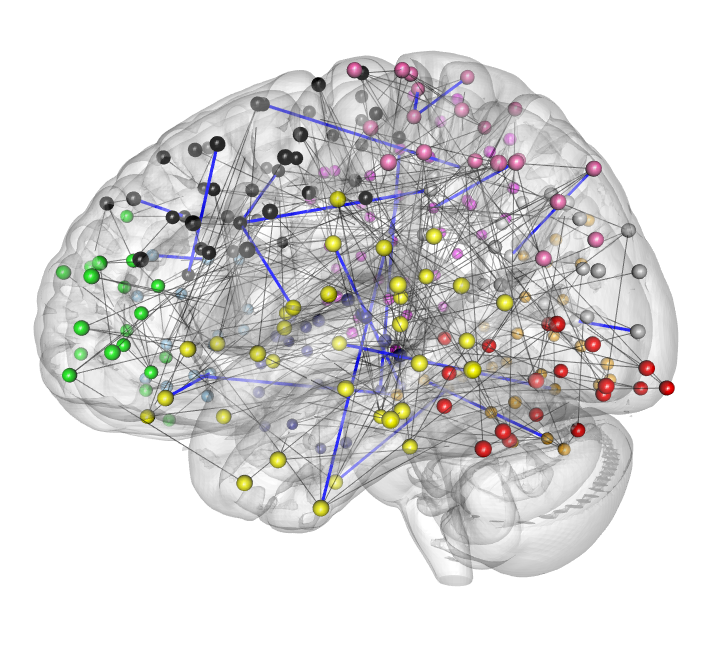} & \includegraphics[trim=0cm 0 0 0, scale=.2, clip]{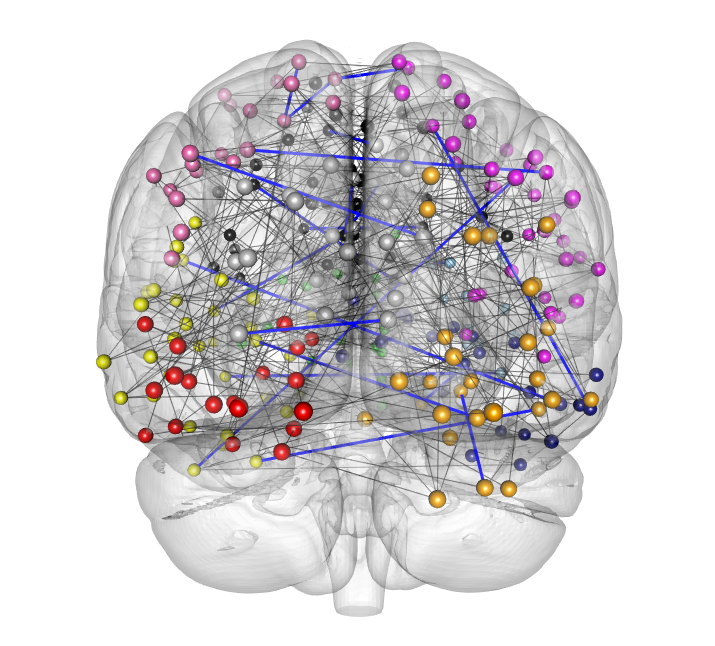}\\
		(a) Health, Transverse & (b) Health, Sagittal & (c) Health, Coronal \vspace{.2in} \\
		\includegraphics[trim=0cm 0 0 0, scale=.2, clip]{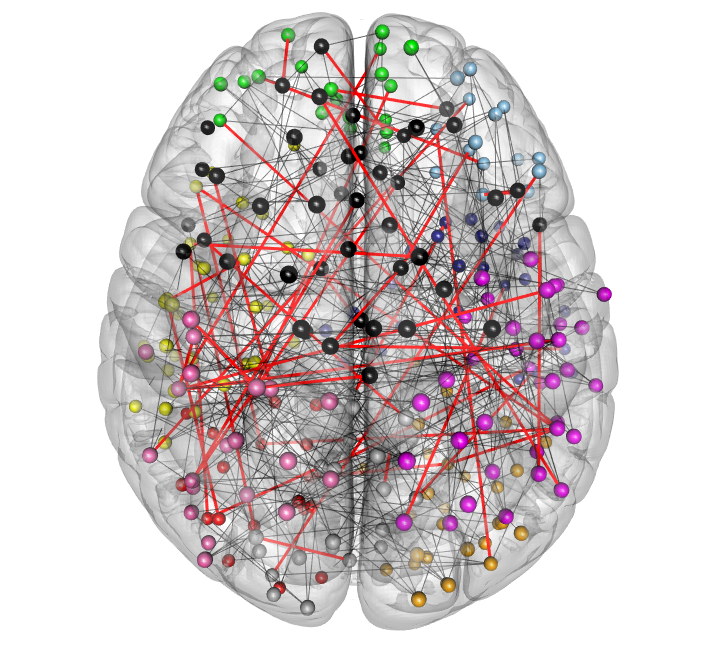} & \includegraphics[trim=0cm 0 0 0, scale=.2, clip]{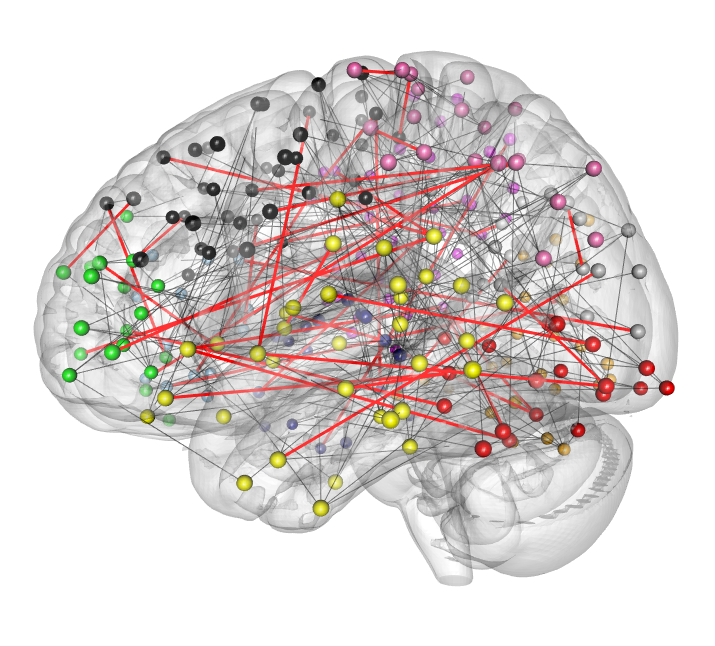} & \includegraphics[trim=0cm 0 0 0, scale=.2, clip]{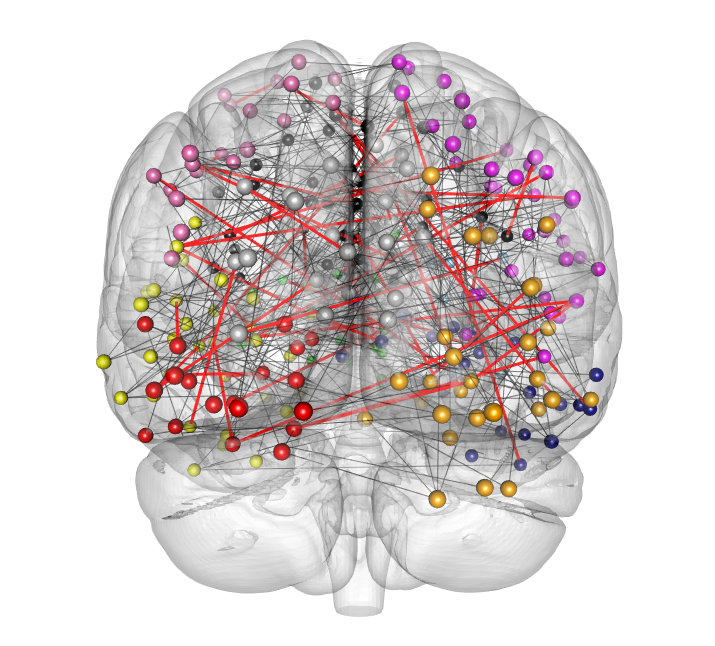}\\
		(a) ADHD, Transverse & (b) ADHD, Sagittal & (c) ADHD, Coronal \vspace{.2in} \\
	\end{tabular}
	\caption{Estimated brain functional connectivity networks using physical locations as covariates to correct heterogeneity. $10$ region clusters are labeled in $10$ colors. Black, blue and red edges represent respectively common edges, unshared edges in the healthy group and in the ADHD group. }
		\label{fig:4}
\end{figure}

The brain network produced by Method 1 is reported in Figures \ref{fig:4}. We omit the networks produced by Methods 2 and 3 since the inferred graphs actually do not differ too much, given that the number of subjects and total sample size are large. All methods give around 90\% identical edges for the two networks, while respectively generating 8.6\%, 8.0\% and 11.6\% unshared edges. Therefore, Methods 1 and 2 with adjustments provide more consistent and trustworthy graphs as batch effect brings more non-biological factors that exaggerate the difference.
Preferences for Methods 1 and 2 should be based on relationship of $p, n_i, m$ and whether the collected covariates are influential enough to explain the loadings. From Figure \ref{fig:4}, it is obvious that the brain is more connected for the ADHD subjects, but the connections are weaker. Actually, the average correlation reduction for the estimated correlation matrices (obtained from $\widehat\Omega^{-1}$) of the two groups is $0.001$; a paired t-test gives p-value $< 2.2\mathrm{e}{-16}$. This is consistent with the recent finding that kids with ADHD show weaker interactions among brain networks \citep{CaiCheSzeSupMen16}. 

\begin{table}
  \centering
  \caption{The degree of unshared edge vertices for each cluster}  \label{tab:Clusterdegree}
  \begin{tabular}{lccccccccccc}
    \hline
     & red  & orange & blue & green & yellow & navy & pink & black & magenta & gray\\
    \hline
    Health & 0 & 5 & 4 & 0 & 6 & 2 & 7 & 9 & 4 & 5\\
    ADHD & 10 & 3 & 8 & 7 & 14 & 5 & 8 & 15 & 10 & 10\\
    p-values ($\%$) & 2.8 & 5.8 & 90.8 & 6.7 & 85.6 & 85.2 & 20.4 & 53.1 & 78.9 & 89.7\\
    \hline
  \end{tabular}
  \end{table}

In addition, we investigate how those unshared edges between two groups of people are distributed across the 10 clusters. We only focus on networks from Method 1. We are interested in which cluster contributes to the difference of the distributions of vertices of the unshared edges the most. We summarized the total degree of unshared edge vertices within each cluster in Table \ref{tab:Clusterdegree}. For each column $j$, we consider the hypothesis testing for $p_{jH} = p_{jD}$, where $H$ and $D$ denotes the healthy and diseased group respectively, meaning that the unshared edges within the $j^{th}$ cluster are found due to the same Bernoulli distribution. A simple chi-square test for the null was carried out for each column and those p-values are reported also in Table \ref{tab:Clusterdegree}. 
The most noteworthy fact is that occipital lobe shows dependency change from right brain (orange) to left brain (red) for ADHD patients. 
The left frontal lobe (green) and the left parietal lobe (pink) have relatively large change in dependence structure compared with other parts of the brain. These are signs that ADHD is a complex disease that affects many regions of the brain. 
The general methodology we provide here could be valuable for further understanding the mechanism of the disease.




\section{Discussions} \label{sec6}

In this paper, we developed a generic method called ALPHA that can consistently estimate and remove data heterogeneity and lead to effective subsequent statistical analysis on the true signal. The entire analysis relies on the pervasive assumption that most of the dimensions are corrupted by the heterogeneous factors. Future work may relax such pervasive conditions to allow for weaker signal batch effect, thus delivering more flexibility to recover the homogeneous residual.


As we have seen, ALPHA is adaptive to factor structures and is flexible to include external information. For brain image data analysis, previous literature rarely took physical locations into considerations. With the new framework, we can take advantage of external characteristics of the voxels or genes relevant to the batch effect, and consistently estimate the pervasive heterogeneity term even with very limited samples. However, this advantage of Projected-PCA is accompanied by more assumptions and the practical issue of selecting proper basis functions and the number of them in sieve approximation. 
On the other hand, if no valuable covariates exist and the sample size is relatively large for each data source, we have shown conventional PCA is still an effective tool. Direct aggregation of less heterogeneous subgroups (say subjects with the same age and gender in the ADHD dataset) might also be helpful to increase the sample size.

Finally, note that after heterogeneity adjustment, the recovered residuals $\check\bU$ are not column-wisely independently distributed anymore. Statistical procedures that require assumptions of i.i.d. data cannot be directly applied on $\check\bU$. However, the ALPHA procedure gives theoretical guarantee for $\|\check\bU-\bU\|_{\max}$ and $\|\widehat\bSigma-\bSigma\|_{\max}$, which serve as foundations for establishing the statistical properties of the subsequent procedure. In this sense, our framework is compatible with any statistical procedure that only requires an accurate estimator as the input, for instance, the CLIME procedure. Methods robust to small perturbations on the truth are preferable.

\bibliographystyle{ims}
\bibliography{WeichenBib}

\appendix
\section*{APPENDIX}
We first give the algorithm for our ALPHA procedure.
We then outline the key ideas of the technical proofs of Theorems \ref{th3.1}---\ref{th3.3} and \ref{th2.3} in respectively the next four sections and leave additional proofs to the technical lemmas in Appendex F.

\section{Algorithm for ALPHA}

\begin{algorithm}
  \caption{Algorithm for adaptive low-rank principal heterogeneity adjustment}
  \label{algo1}
  \renewcommand{\algorithmicrequire}{\textbf{\underline{Input}:}}
  \renewcommand{\algorithmicensure}{\textbf{\underline{Output}:}}
  $\;$ \algorithmicrequire{$\;$Panel $\bX^i_{p\times n_i}$ and $d$-dimensional $\{\bW_j^i\}_{j=1}^p$ from $m$ data sources}

  $\;$ \algorithmicensure{$\;\check{\bU^i}$, the adjusted estimator for $\bU^i$ and $\widehat\bSigma$}

  \begin{algorithmic}[1]
	\Procedure{ALPHA}{}
	\For{each subject $i \le m$}
		\State $\widehat {K^i} \leftarrow$ non-projected eigenvalue-ratio method
		\State test $H_0: \bG^i(\bX^i) = 0$ by $S^i$
	\EndFor
	\State $\mathcal M_1, \mathcal M_2 \leftarrow$ control FDR by Benjamini-Hochberg
	\For{each subject $i \le m$}
		\If{ $i \in \mathcal M_1$ }
    			\State $\widehat{\bF^i}/\sqrt{n_i} \leftarrow$ eigenvectors of ${\bX^i}' \bX^i$ corresponding to top $\widehat{K^i}$ eigenvalues
			\State  $\widehat{\bLambda^i} \leftarrow \bX^i \widehat{\bF^i}/n_i$, $\widehat{\bU^i} \leftarrow \bX^i - \widehat{\bLambda^i}\widehat{\bF^i}$ and $\check{\bU^i} \leftarrow \widehat{\bU^i}$
		\Else
			\State $\widetilde {K^i} \leftarrow$ projected eigenvalue-ratio method
			\State $\bP^i \leftarrow \Phi(\bW^i)(\Phi(\bW^i)'\Phi(\bW^i))^{-1}\Phi(\bW^i)'$
			\State $\widetilde{\bF^i}/\sqrt{n_i} \leftarrow$ eigenvectors of ${\bX^i}' \bP^i \bX^i$ corresponding to top $\widetilde{K^i}$ eigenvalues
			\State $\widetilde{\bLambda^i} \leftarrow \bX^i \widetilde{\bF^i}/n_i$, $\widetilde{\bU^i} \leftarrow \bX^i - \widetilde{\bLambda^i}\widetilde{\bF^i}$ and $\check{\bU^i} \leftarrow \widetilde{\bU^i}$
		\EndIf
	\EndFor
	\State $\widehat\bSigma \leftarrow (\sum_i n_i - \sum_i {\widehat {K_i}})^{-1} \sum_{i=1}^m {\check{\bU^i}} {\check {\bU^i}}'$
	\State \Return $\{\check{\bU^i}\}_{i=1}^m$ and $\widehat\bSigma$
	\EndProcedure
  \end{algorithmic}
\end{algorithm}

\section{Proof of Theorem \ref{th3.1}}

\begin{proof}
By definition of $\check \bU$, $\check\bU = \bU(\bI - n^{-1}\bF\bF') + n^{-1} \bX (\check\bF\check\bF' - \bF\bF')$. We first look at the converge of $\check\bU - \bU$. Obviously $\bPi = n^{-1} \bX (\check\bF\check\bF' - \bF\bF') = I + II$ where
$$
I = \frac1n \bLambda\bF' (\check\bF\check\bF' - \bF\bF'), \;\; II = \frac1n \bU (\check\bF\check\bF' - \bF\bF')\,.
$$
Since $\bF' (\check\bF\check\bF' - \bF\bF') = \bF' (\check\bF - \bF\bH) \check\bF'  + n \bH (\check\bF - \bF\bH)' + n (\bH\bH' - \bI)\bF'$, we have
$$
\|I\|_{\max} = O_P(\|\bLambda\|_{\max} (\|\bF' (\check\bF - \bF\bH)\|_{\max} \|\check\bF/n\|_{\max} + \|\check\bF - \bF\bH\|_{\max} + \|\bH\bH' - \bI\|_{\max} \|\bF\|_{\max}  ) )\,.
$$
Similarly $\bU (\check\bF\check\bF' - \bF\bF') = \bU (\check\bF - \bF\bH) \check\bF' + \bU\bF \bH(\check\bF - \bF\bH)' + \bU\bF (\bH\bH' - \bI) \bF'$, so
$$
\|II\|_{\max} = O_P(\|\bU' (\check\bF - \bF\bH)\|_{\max} \|\check\bF/n\|_{\max} + \|\bU\bF/n\|_{\max} (\|\check\bF - \bF\bH\|_{\max} + \|\bH\bH' - \bI\|_{\max} \|\bF\|_{\max} ) )\,.
$$
According to Lemma \ref{tla.2} (i), $\|\bU\bF/n\|_{\max} = O_P(1)$ and noting both $\|\bF\|_{\max}$ and $\|\check\bF\|_{\max}$ are $O_P(\sqrt{n})$, we conclude the result for $\|\bPi\|_{\max}$ easily.

Now we consider $\check\bU\check\bU'$ in the following.
\begin{align*}
\check\bU\check\bU'& = \bU (\bI - n^{-1} \bF\bF') \bU' + n^{-1} \bU(\bI - n^{-1}\bF\bF')(\check\bF\check\bF' - \bF\bF')\bX'+ n^{-2}\bX(\check\bF\check\bF' - \bF\bF')^2\bX' \\
&=: \bU \bU' - \frac{1}{n} \bU \bF\bF' \bU' + III + IV\,.
\end{align*}
So $\bDelta = III + IV$ and it suffices to bound the two terms.
\begin{align*}
\|III\|_{\max} & = O_P(\|n^{-1} \bU(\bI - \bF\bF'/n)\check\bF\check\bF' \bF\|_{\max} \|\bLambda\|_{\max} + \|n^{-1} \bU(\bI - \bF\bF'/n)\check\bF\check\bF' \bU'\|_{\max}) \\
& =: O_P(\|\bJ_1\|_{\max}\|\bLambda\|_{\max} + \|\bJ_2\|_{\max})\,.
\end{align*}
Decompose $\bJ_1$ by $\bJ_1 = n^{-1}\bU(\check\bF - \bF\bH)\check\bF'\bF - n^{-2}\bU \bF\cdot \bF'(\check\bF-\bF\bH)\check\bF'\bF$. Therefore,
$$
\|\bJ_1\|_{\max} = O_P(\|\bU(\check\bF - \bF\bH)\|_{\max} + n^{-1} \|\bU \bF\|_{\max} \|\bF'(\check\bF-\bF\bH)\|_{\max})\,,
$$
since $\|\check\bF'\bF/n\|_{\max} \le \|\check\bF'\bF/n\|_{F} \le \|\check\bF'\|_F \|\bF\|_F/n = K$. Similar to $\bJ_1$, we decompose $\bJ_2$ only replacing $\check\bF'\bF$ with $\check\bF'\bU'$. According to Lemma \ref{tla.2} (i), $\|\check\bF'\bU'/n\|_{\max} = O_P( \|\bU\bF/n\|_{\max} + \|\bU(\check\bF - \bF\bH)\|_{\max} )= O_P(1 + \|\bU(\check\bF - \bF\bH)\|_{\max})$, hence $\|\bJ_2\|_{\max} = O_P(\|\bJ_1\|_{\max}(1 + \|\bU(\check\bF - \bF\bH)\|_{\max}))$. We then conclude that $\|III\|_{\max} = O_P((\|\bU(\check\bF - \bF\bH)\|_{\max} + n^{-1} \|\bU \bF\|_{\max} \|\bF'(\check\bF-\bF\bH)\|_{\max})(\|\bLambda\|_{\max} + \|\bU(\check\bF - \bF\bH)\|_{\max}))$.

Now let us take a look at $IV$. $\|IV\|_{\max} = \|\bD_1 + \bD_2 + \bD_2' + \bD_3\|_{max}$ where
\begin{align*}
\bD_1 & = n^{-2} \bLambda\bF' (\check\bF\check\bF' - \bF\bF')^2 \bF \bLambda' = \bLambda (n \bI - n^{-1}\bF'\check\bF\check\bF'\bF)\bLambda' \,, \\
\bD_2 & = n^{-2} \bU  (\check\bF\check\bF' - \bF\bF')^2 \bF \bLambda'  = - n^{-2} \bU \bF\bF'(\check\bF\check\bF' - \bF\bF')\bF \bLambda' \,\\
\bD_3 & = n^{-2} \bU  (\check\bF\check\bF' - \bF\bF')^2 \bU' \,.
\end{align*}
By assumption, $\|\bH\|_{\max} \le \|\bH\| = O_P(1)$. Simple decompositions of $\bD_1$ gives
$$
\|\bD_1\|_{\max} = O_P((\|\bF'(\check\bF - \bF\bH)\|_{\max} + n\|\bH\bH' - \bI\|_{\max})\|\bLambda\|_{\max}^2)\,.
$$
Since $\bD_2 = -n^{-2} \bU\bF\bF'(\check\bF - \bF\bH)\check\bF'\bF\bLambda' - n^{-1}\bU\bF\bH(\check\bF-\bF\bH)'\bF\bLambda' - \bU\bF(\bH\bH' - \bI)\bLambda'$, we have
$$
\|\bD_2\|_{\max} = O_P(\|\bU\bF/n\|_{\max} \|\bD_1\|_{\max}) = O_P(\|\bD_1\|_{\max})\,.
$$
It is also not hard to show  $\|\bD_3\|_{\max} = O_P(\|III\|_{\max} + \|\bD_1\|_{\max})$. Under both Theorems \ref{th3.2} and \ref{th3.3} (replacing $\check\bF$ by $\widehat\bF$ for regime 1 and $\widetilde\bF$ for regime 2), we can check the following relationship holds:
$$
n^{-1} \|\bU\bF\|_{\max} \|\bU(\check\bF - \bF\bH) \|_{\max} = O_P(\|\bLambda\|_{\max}^2).
$$
Therefore we have
\begin{align*}
\|\bDelta\|_{\max} = \|III + IV\|_{\max} = & O_P(\|\bU(\check\bF - \bF\bH)\|_{\max} \|\bLambda\|_{\max} + \|\bU(\check\bF - \bF\bH)\|_{\max}^2  \\
& + \|\bF'(\check\bF - \bF\bH)\|_{\max}\|\bLambda\|_{\max}^2 + n\|\bH\bH' - \bI\|_{\max}\|\bLambda\|_{\max}^2)\,.
\end{align*}
\end{proof}

\section{Proof of Theorem \ref{th3.2}}
\subsection{Convergence of factors $\hF$}
Let $\bK$ denote the $K \times K$ diagonal matrix consisting of the first $K$ largest eigenvalues of $(pn)^{-1}\bX'\bX$ in descending order.
By the definition of eigenvalues, we have
 $$
 \frac{1}{np}(\bX'\bX)\hF =\hF\bK\,.
 $$

Recall $\bH=(np)^{-1}\bLambda'\bLambda\bF'\hF\bK^{-1}$.
Substituting $\bX=\bLambda\bF' + \bU,$
we have,
\begin{equation}\label{eqa.1}
\hF-\bF\bH= \Big( \sum_{i=1}^{3}\bE_i \Big) \bK^{-1}\,,
\end{equation}
$$
\bE_1 = \frac{1}{np}\bF\bLambda'\bU\hF,\quad
\bE_2= \frac{1}{np} \bU' \bLambda\bF'\hF, \quad \bE_3=\frac{1}{np}\bU'\bU\hF\,.
$$

To bound  $\|\hF-\bF\bH\|_{\max}$, note that there is a constant $C>0$, so that
$$
\|\hF-\bF\bH\|_{\max} \leq C \|\bK^{-1}\|_2 \sum_{i=1}^{3} \|\bE_i\|_{\max}.
$$
Hence we need to bound $\|\bE_i\|_{\max}$ for $i=1,2,3$ since $\|\bK^{-1}\|_2 = O_P(1)$. The following lemma gives the stochastic bounds for each individual term.

\begin{lem}\label{laC.1}
(i) $\|\bE_1\|_{F}=O_P(\sqrt{n/p}) = \|\bE_2\|_{F}$, $\|\bE_3\|_{F}=O_P(1/\sqrt{n} + 1/\sqrt{p} + \sqrt{n}/p)$\,. \\
(ii) $\|\bE_1\|_{\max}=O_P(\sqrt{\log n/p}) = \|\bE_2\|_{\max}$, $\|\bE_3\|_{\max}=O_P(1/\sqrt{p} + \sqrt{\log n}/n)$\,.
\end{lem}

\begin{proof}
(i) Obviously $\|\bE_1\|_{F} \le p^{-1} \|\bLambda'\bU\|_F = O_P(\sqrt{n/p})$ according to Lemma \ref{tla.PCA1}. $\|\bE_2\|_F$ attains the same rate. In addition, $\|\bE_3\|_F \le n^{-1/2}p^{-1} \|\bU'\bU\|_F = O_P(1+\sqrt{n/p})$ again according to Lemma \ref{tla.PCA1}. So combining the three terms, we have $\|\hF -\bF\bH\|_F = O_P(1+\sqrt{n/p})$. We now refine the bound for $\|\bE_3\|_F$. $\|\bE_3\|_F \le (np)^{-1} (\|\bU'\bU\bF\|_F\|\bH\|_F + \|\bU'\bU\|_F \|\hF - \bF\bH\|_F) = O_P(1/\sqrt{n} + 1/\sqrt{p} + \sqrt{n}/p)$. Then the refined rate of $\|\hF -\bF\bH\|_F$ is $O_P(\sqrt{n/p} + 1/\sqrt{n})$.

(ii) Since $ \|\bLambda'\bU\hF\|_F = O_P(n\sqrt{p})$ by Lemma \ref{tla.PCA1},
$$
\|\bE_1\|_{\max} = O_P((np)^{-1} \|\bF\|_{\max} \|\bLambda'\bU\hF\|_F) = O_P(\sqrt{\log n/p})\,.
$$
$\|\bE_2\|_{\max}$ is bounded by $p^{-1} \|\bU'\bLambda\|_{\max} = O_P(\sqrt{\log n/p})$ while $\|\bE_3\|_{\max}$ is bounded by
$$
O_P\Big((np)^{-1}(\|\bU'\bU\bF\|_{\max} + \sqrt{n}\|\bU'\bU\|_{\max}\|\hF - \bF\bH\|_{F})\Big)\,,
$$
which based on results of Lemma \ref{tla.PCA2} and (i) is $O_P(1/\sqrt{p} + \sqrt{\log n}/n)$.
\end{proof}

The final rate of convergence for $\|\hF-\bF\bH\|_{\max}$ and $\|\hF-\bF\bH\|_{F}$ are summarized as follows.

\begin{prop} \label{paC.1}
\begin{equation}
\|\hF-\bF\bH\|_{\max} =O_P\Big(\sqrt{\frac{\log n}{p}} + \frac{\sqrt{\log n}}{n} \Big) \;\; \text{and} \;\; \|\hF-\bF\bH\|_{F} =O_P\Big(\sqrt{\frac np} + \frac{1}{\sqrt{n}}\Big).
\end{equation}
\end{prop}

\begin{proof} The results follow from Lemmas \ref{laC.1}.
\end{proof}

\subsection{Rates of $\|\bF'(\hF-\bF\bH)\|_{\max}$ and $\|\bH\bH' - \bI\|_{\max}$}
Note first that the two matrices under consideration is both $K$ by $K$, so we do not lose rates bounding them by their Frobenius norm.

Let us find out rate for $\|\bF'(\hF-\bF\bH)\|_{F}$. Basically we need to bound $\|\bF'\bE_i\|_F$ for $i = 1,2,3$. Firstly
$$
\|\bF'\bE_1\|_{F} = p^{-1}\|\bLambda'\bU\hF\|_F \le p^{-1}(\|\bLambda'\bU\bF\|_F \|\bH\|_F + \|\bLambda'\bU\|_F \|\hF-\bF\bH\|_F)\,.
$$
Since $\|\bLambda'\bU\bF\|_F = O_P(\sqrt{np})$ and $\|\bLambda'\bU\|_F = O_P(\sqrt{np})$ by Lemma \ref{tla.PCA1}, we have $\|\bF'\bE_1\|_{F} = O_P(\sqrt{n/p} + n/p)$. Secondly,
$$
\|\bF'\bE_2\|_{F} \le p^{-1}\|\bF'\bU'\bLambda\|_F = O_P(\sqrt{n/p})\,.
$$
Finally,
$$
\|\bF'\bE_3\|_{F} = O_P\Big(\frac{1}{np}\|\bU\bF\|_F^2 + \frac{1}{np} \|\bF'\bU'\bU\|_F\|\hF - \bF\bH\|_F\Big) = O_P(1+\sqrt{n}/p)\,.
$$
So combining three terms we have $\|\bF'(\hF-\bF\bH)\|_{\max} \le \|\bF'(\hF-\bF\bH)\|_{F} = O_P(1+\sqrt{n/p})$.

Now we bound $\|\bH\bH' - \bI\|_{F}$. Since $\bH'\bH = n^{-1} (\bF\bH - \hF)'\bF\bH + n^{-1} \hF'(\bF\bH - \hF) + \bI$, we have
$$
\|\bH'\bH - \bI \|_F = O_P(\frac{1}{n} \|\bF'(\hF-\bF\bH) \|_F + \frac1n \|\hF - \bF\bH\|_F^2) = O_P\Big(\frac1n + \frac1p\Big)\,.
$$
Therefore $\|\bH\bH' - \bI\|_{F}$ has the same rate since $\|\bH\bH' - \bI\|_{F} \le \|\bH\|_F \|\bH'\bH - \bI\|_{F} \|\bH^{-1}\|_F$. So $\|\bH\bH' - \bI\|_{\max} = O_P(1/n + 1/p)$.

\subsection{Rate of $\|\bU(\hF-\bF\bH)\|_{\max}$}

In order to study rate of $\|\bU(\hF-\bF\bH)\|_{\max}$, we essentially need to bound $\|\bU\bE_i\|_{\max}$ for $i=1,2,3$. We handle each term separately.
$$
\|\bU\bE_1\|_{\max} = O_P(\frac {1}{np} \|\bU\bF\|_{\max} \|\bLambda'\bU\hF\|_{F})  = O_P(\frac {1}{n} \|\bU\bF\|_{\max} \|\bF'\bE_1\|_F) = O_P\Big(\sqrt{\frac{\log p}{p}} + \frac{\sqrt{n \log p}}{p}\Big)\,.
$$
By Lemma \ref{tla.5}, $\|\bU\bU'\bLambda\|_{\max} = O_P(\sqrt{np \log p} + n\|\bSigma\|_1)$. Therefore,
$$
\|\bU\bE_2\|_{\max} = O_P(\frac {1}{p} \|\bU\bU'\bLambda\|_{\max})  = O_P\Big(\frac{n\|\bSigma\|_1}{p}  + \sqrt{\frac{n \log p}{p}}\Big)\,.
$$
From bounding $\|\bE_3\|_F$, the last term has rate
$$
\|\bU\bE_3\|_{\max} = \frac{1}{np} \|\bU\bU'\bU\hF\|_{\max} \le \frac{1}{\sqrt{n}p} \|\bU\|_{\max} \|\bU'\bU\hF\|_F
= O_P((1+n/p)\sqrt{\log p})\,.
$$
So combining three terms, we conclude
$\|\bU(\hF-\bF\bH)\|_{\max} = O_P((1+n/p)\sqrt{\log p} + n\|\bSigma\|_1/p)$.

\section{Proof of Theorem \ref{th3.3}}
\subsection{Convergence of factors $\widetilde\bF$}
Let $\bK$ denote the $K \times K$ diagonal matrix consisting of the first $K$ largest eigenvalues of $(pn)^{-1}\bX'\bP\bX$ in descending order.
By the definition of eigenvalues, we have
 $$
 \frac{1}{np}(\bX'\bP\bX)\widetilde\bF =\widetilde\bF\b\,.
 $$

Recall $\bH=(np)^{-1}\bB'\Phi(\bW)'\Phi(\bW)\bB\bF'\tF\bK^{-1}$.
Substituting $
\bX=\Phi(\bW)\bB\bF' + \bR(\bW)\bF' + \bGamma\bF' + \bU,
 $
we have,
\begin{equation}\label{eqa.1}
\tF-\bF\bH= \Big( \sum_{i=1}^{15}\bA_i \Big) \bK^{-1}
\end{equation}
where $\bA_i, i \le 3$ has nothing to do with $\bR(\bW)$ and $\bGamma$:
$$
\bA_1 = \frac{1}{np}\bF\bB'\Phi(\bW)'\bU\tF,\quad
\bA_2= \frac{1}{np} \bU' \Phi(\bW)\bB\bF'\tF, \quad \bA_3=\frac{1}{np}\bU'\bP\bU\tF\,;
$$
$\bA_i, 3 \le i \le 8$ takes care of terms involving $\bR(\bW)$:
\begin{eqnarray*}
\bA_4&=& \frac{1}{np} \bF\bB'\Phi(\bW)' \bR(\bW)\bF'\tF,\quad
\bA_5 =\frac{1}{np} \bF\bR(\bW)' \Phi(\bW)\bB\bF' \tF,
\cr
\bA_6 &=& \frac{1}{np} \bF\bR(\bW)' \bP \bR(\bW)\bF' \tF,\quad \bA_7 = \frac{1}{np} \bF\bR(\bW)' \bP \bU \tF,
\quad
\bA_8= \frac{1}{np} \bU' \bP \bR(\bW)\bF'\tF\,;
\end{eqnarray*}
the remaining are terms involving $\bGamma$:
\begin{eqnarray*}
\bA_9&=& \frac{1}{np} \bF\bB'\Phi(\bW)' \bGamma\bF'\tF,\quad
\bA_{10} =\frac{1}{np} \bF\bGamma' \Phi(\bW)\bB\bF' \tF, \quad
\bA_{11} = \frac{1}{np} \bF\bGamma' \bP \bGamma \bF' \tF,
\cr
\bA_{12} &=& \frac{1}{np} \bF\bGamma' \bP \bU \tF, \quad
\bA_{13} = \frac{1}{np} \bU' \bP \bGamma\bF'\tF, \quad
\bA_{14} = \frac{1}{np} \bF\bR' \bP\bGamma\bF'\tF,\quad
\bA_{15} = \frac{1}{np} \bF\bGamma'\bP\bR\bF'\tF.
\end{eqnarray*}

To bound  $\|\tF-\bF\bH\|_{\max}$, as in Theorem \ref{th3.2} we only need to bound $\|\bA_i\|_{\max}$ for $i=1,...,15$ since again we have $\|\bK^{-1}\|_2 = O_P(1)$. The following lemma gives the rate for each term.

\begin{lem}\label{la.2}
(i) $\|\bA_1\|_{\max}=O_P(\sqrt{\log n/p}) = \|\bA_2\|_{\max}$,\\
(ii) $\|\bA_3\|_{\max}=O_P(J\phi_{\max} \sqrt{\log(nJ)}/p)$, \\
(iii) $\|\bA_4\|_{\max}=O_P(J^{-\kappa/2}\sqrt{\log n} ) = \|\bA_5\|_{\max}$ and $\|\bA_9\|_{\max}=O_P(\sqrt{\nu_p\log n/p} ) = \|\bA_{10}\|_{\max}$, \\
(iv) $\|\bA_6\|_{\max}=O_P(J^{-\kappa}\sqrt{\log n})$ and $\|\bA_{11}\|_{\max} = O_P(J\nu_p\sqrt{\log n}/p)$, \\
(v) $\|\bA_7\|_{\max}=O_P(\phi_{\max}\sqrt{p^{-1} J^{1-\kappa} \log(nJ) \log n}) = \|\bA_8\|_{\max}$ \\and $\|\bA_{12}\|_{\max} = O_P(J\phi_{\max}\sqrt{\nu_p \log (nJ) \log n}/p) = \|\bA_{13}\|_{\max}$, \\
(vi) $\|\bA_{14}\|_{\max} = O_P(\sqrt{p^{-1} J^{1-\kappa} \nu_p\log n}) = \|\bA_{15}\|_{\max}$.
\end{lem}

\begin{proof}(i) Because $\|\bF\|_{\max}=O_P(\sqrt{\log n})$, $\|\tF\|_F=O_P(\sqrt{n})$. By  Lemmas \ref{tla.1} and \ref{tla.2},  $\| \bU' \Phi(\bW)\bB\|_{F} = O_P(\sqrt{pn})$ and $\| \bU' \Phi(\bW)\bB\|_{\max} = O_P(\sqrt{p \log n})$. Hence
$$
\|\bA_1\|_{\max} \le \frac{\sqrt{K}}{np} \|\bF\|_{\max} \|\bB'\Phi(\bW)'\bU\|_{F} \|\tF\|_{F} = O_P(\sqrt{\log n/p}),
$$
$$
\|\bA_2\|_{\max} \le \frac{\sqrt{K}}{np} \|\bU' \Phi(\bW)\bB\|_{\max} \|\bF\|_{F} \|\tF\|_{F} = O_P(\sqrt{\log n/p}).
$$

(ii) We have $\bA_3=\frac{1}{np}\bU'\Phi(\bW)(\Phi(\bW)'\Phi(\bW))^{-1}\Phi(\bW)'\bU\tF$. By Lemma \ref{tla.1} and \ref{tla.2}, $\|\bU'\Phi(\bW)\|_{F}=O_P(\sqrt{npJ})$ and $\|\bU'\Phi(\bW)\|_{\max}=O_P(\phi_{\max}\sqrt{p \log(nJ)})$. By Assumption \ref{ass2.3}, $\|(\Phi(\bW)'\Phi(\bW))^{-1}\|_2=O_P(p^{-1})$. Note the fact that for matrix $\bA_{m \times n}$, $\bB_{n\times n}$, $\bC_{n\times r}$, $\|\bA\bB\bC\|_{\max} = \max_{i \le m, k \le r} |\ba_i' \bB \bc_k| \le \sqrt{n}  \|\bA\|_{\max} \|\bB\|_2 \|\bC\|_{F}$. So
\begin{align*}
 \|\bA_3\|_{\max} & \le \frac{\sqrt{Jd}}{np} \|\bU'\Phi(\bW)\|_{\max} \|(\Phi(\bW)'\Phi(\bW))^{-1}\|_2 \|\Phi(\bW)'\bU\|_F \|\tF\|_{F} \\
 &=O_P(J\phi_{\max}\sqrt{\log(nJ)}/p).
 \end{align*}

(iii)  Note that $\|\Phi(\bW)\bB\|_2\leq\|\bG(\bW)\|_2+\|\bR(\bW)\|_2=O_P(\sqrt{p})$, and $\|\bR(\bW)\|_{\max}=O_P(J^{-\kappa/2})$. Hence we have $\|\bB' \Phi(\bW)' \bR(\bW)\|_{\max} \le \|\bB' \Phi(\bW)'\|_1 \|\bR(\bW)\|_{\max} \le \sqrt{p} \|\bB' \Phi(\bW)'\|_2 \|\bR(\bW)\|_{\max} = O_p(pJ^{-\kappa/2})$. Thus
$$
\|\bA_4\|_{\max} \leq\frac{K^{3/2}}{n p}\|\bF\|_{\max} \|\bB' \Phi(\bW)' \bR(\bW)\|_{\max} \|\bF\tF\|_{F} =O_P(J^{-\kappa/2}\sqrt{\log n}).
$$
Similarly, $\|\bA_5\|_{\max}$ attains the same rate of convergence.

In addition, notice $\bA_9, \bA_{10}$ have similar representation as $\bA_4, \bA_5$. The only difference is to replace $\bR$ by $\bGamma$. It is not hard to see $\|\bB'\Phi'\bGamma\|_{\max} = O_P(\sqrt{p \nu_p})$. Therefore $\|\bA_9\|_{\max}=O_P(\sqrt{\nu_p\log n/p} ) = \|\bA_{10}\|_{\max}$.

(iv) Note that $\|\bP\|_2=\|(\Phi(\bW)'\Phi(\bW))^{-1/2}\Phi(\bW)'\Phi(\bW)(\Phi(\bW)'\Phi(\bW))^{-1/2}\|_2=1$ and $\|\bR(\bW)'\bP \bR(\bW)\|_{\max} \le p \|\bR(\bW)\|_{\max}^2 \|\bP\|_2 = O_p(pJ^{-\kappa})$. Hence
$$
\|\bA_6\|_{\max} \le \frac{K}{np} \|\bF\|_{\max} \|\bR(\bW)'\bP \bR(\bW)\|_{\max} \|\bF\tF\|_{F}=O_P(J^{-\kappa}\sqrt{\log n}).
$$

$\bA_{11}$ has similar representation as $\bA_6$. Since $\|\bGamma'\bP\bGamma\|_{\max} \le \|\Phi'\bGamma\|_F^2 \|(\Phi'\Phi)^{-1}\|_2 = O_P(J\nu_p)$, we have $\|\bA_{11}\|_{\max} = O_P(J\nu_p\sqrt{\log n}/p)$.

(v) According to Lemma \ref{tla.2}, $\|\bU'\Phi(\bW)\|_{\max} = O_P(\phi_{\max}\sqrt{p \log (nJ) })$. Thus
\begin{align*}
\|\bA_7\|_{\max} & \le \frac{K}{\sqrt{n}p}\|\bF\|_{\max} \|\tF\|_{F} \|\bR' \Phi (\Phi'\Phi)^{-1} \Phi '\bU\|_{\max} \\
& \le O_p(p^{-1}\sqrt{J\log n}) \|\bR'\Phi\|_{F} \|(\Phi'\Phi)^{-1}\|_2 \|\Phi'\bU\|_{\max} = O_p\Big(\phi_{\max}\sqrt{\frac{J \log(nJ) \log n}{pJ^\kappa}}\Big)\,,
\end{align*}
since $\|\bR'\Phi\|_{F} \le \|\bR\|_F \|\Phi\|_2 = O_P(pJ^{-\kappa/2})$.
The rate of convergence for $\bA_8$ can be bounded in the same way. So do $\bA_{12}$ and $\bA_{13}$. Given that $\|\bGamma'\Phi\|_F = O_P(pJ\nu_p)$, we have $\|\bA_{12}\|_{\max} = O_P(J\phi_{\max}\sqrt{\nu_p \log (nJ) \log n}/p) = \|\bA_{13}\|_{\max}$.

(vi) Obviously, $\|\bA_{14}\|_{\max} = O_P(p^{-1}\sqrt{\log n}\|\bR'\bP\bGamma\|_{\max})$ and $ \|\bR'\bP\bGamma\|_{\max} \le \|\bR'\Phi\|_F \|(\Phi'\Phi)^{-1}\| \|\Phi'\bGamma\|_{F}$. We conclude $\|\bA_{14}\|_{\max} = O_P(\sqrt{p^{-1} J^{1-\kappa} \nu_p\log n})$. Same bound holds for $\bA_{15}$.
\end{proof}

The final rate of convergence for $\|\tF-\bF\bH\|_{\max}$ and  $\|\tF-\bF\bH\|_{F}$ are summarized as follows.

\begin{prop} \label{pa.1} Choose $J = (p\min(n,p,\nu_p^{-1}))^{1/\kappa}$ and assume $J^2\phi_{\max}^2 \log (nJ)=O(p)$ and $\nu_p = O(1)$,
\begin{equation}
\|\tF-\bF\bH\|_{\max} =O_P\Big(\sqrt{\frac{\log n}{p}}\Big) \;\; \text{and} \;\; \|\tF-\bF\bH\|_{F} =O_P\Big(\sqrt{\frac{n}{p}}\Big).
\end{equation}
\end{prop}

\begin{proof} The max norm result follows from Lemmas \ref{la.2} and (\ref{eqa.1}), while the Frobenius norm result has been shown in \cite{FLW14}.
\end{proof}

\subsection{Rates of $\|\bF'(\tF-\bF\bH)\|_{\max}$ and $\|\bH\bH' - \bI\|_{\max}$}

Note first that the two matrices under consideration is both $K$ by $K$, so we do not lose rates bounding them by their Frobenius norm.

It has been proved in \cite{FLW14} that $\|\bF'(\tF-\bF\bH)\|_{F} = O_P(\sqrt{n/p} + n/p + n\sqrt{\nu_p/p} + nJ^{-\kappa/2})$. By the choice of $J$, the last term vanishes. So
$$
\|\bF'(\tF-\bF\bH)\|_{\max} \le \|\bF'(\tF-\bF\bH)\|_{F} =  O_P(\sqrt{n/p} + n/p + n\sqrt{\nu_p/p}).
$$

\cite{FLW14} also showed that $\|\bH'\bH - \bI\|_{F} = O_P(1/p+1/\sqrt{pn} + J^{-\kappa/2} + \sqrt{\nu_p/p})$. Since $\|\bH\|$ and $\|\bH^{-1}\|$ are both $O_P(1)$, we easily show $\|\bH\bH' - \bI\|_{\max} \le \|\bH\bH' - \bI\|_{F} \le \|\bH\| \|\bH'\bH - \bI\|_{F} \|\bH^{-1}\| = O_P(1/p+1/\sqrt{pn} + \sqrt{\nu_p/p})$ since $J^{\kappa} \ge p/\nu_p$.

\subsection{Rate of $\|\bU(\tF-\bF\bH)\|_{\max}$}

By (\ref{eqa.1}), in order to bound $\|\bU(\tF-\bF\bH)\|_{\max}$, we essentially need to bound $\|\bU\bA_i\|_{\max}$ for $i = 1,\dots,15$. We do not bother going into the details of each term again as in Lemma \ref{la.2}. However, we point out the difference here. All $\bA_i$ are separated into two types: the ones starting with $\bF$ and the ones starting with $\bU$.

If a term $\bA_i$ starts with $\bF$, say $\bA_i = \bF \bQ$, in Lemma \ref{la.2}, we bound $\|\bA_i\|_{\max}$ using $\sqrt{K} \|\bF\|_{\max} \|\bQ\|_{F}$. Now we use bound $\|\bU\bA_i\|_{\max} \le \sqrt{K} \|\bU\bF\|_{\max} \|\bQ\|_{F}$ so that we obtain all related rates by just changing rate $\|\bF\|_{\max} = O_P(\sqrt{\log n})$ to $\|\bU\bF\|_{\max} = O_P(\sqrt{n \log p})$.

Terms starting with $\bU$ includes $\bA_i, i = 2, 3, 8, 13$. In Lemma \ref{la.2}, we bound $\|\bA_i\|_{\max}, i = 3, 8, 13$ using $\|\bU'\Phi\|_{\max}$ while we bound $\|\bA_2\|_{\max}$ using $\|\bU'\Phi\bB\|_{\max}$. Correspondingly now we need to control $\|\bU\bU'\Phi\|_{\max}$ and $\|\bU\bU'\Phi\bB\|_{\max}$ separately to update the rates. The derivation is relegated to Lemma \ref{tla.5}. We have $\|\bU\bU'\Phi(\bW)\|_{\max}=O_P(\phi_{\max}(\sqrt{np\log p} + n\|\bSigma\|_1))$ and $\|\bU\bU'\Phi(\bW)\bB\|_{\max}=O_P(\sqrt{np\log p} + nJ\phi_{\max} \|\bSigma\|_1)$.

So we replace the corresponding terms in Lemma \ref{la.2}. It is not hard to see the dominating term is $\|\bU\bA_2\|_{\max} = O_P(\sqrt{n \log p/p} + nJ\phi_{\max}\|\bSigma\|_1/p)$. Therefore, $\|\bU(\tF-\bF\bH)\|_{\max}$ has the same rate.

\section{Proof of Theorem \ref{th2.3}}
\begin{proof}
Denote the empirical covariance matrix as
$$
\bSigma_N = \frac{1}{N} \sum_{i=1}^m \bU^i{\bU^i}'.
$$
As in \cite{CLL11}, the upper bound on $\|\widehat\bOmega - \bOmega\|$ is obtained by proving
\begin{equation} \label{eqa.6}
\|(\widehat\bSigma - \bSigma_N)\bOmega\|_{\max} = O_p(\tau_{N, p}) \;\; \text{and} \;\; \|(\bSigma_N - \bSigma) \bOmega\|_{\max} = O_p(\tau_{N,p}).
\end{equation}
Once the two bounds are established, we proceed by observing
$$
\|\bI_p - \widehat\bSigma \bOmega\|_{\max} = \|(\widehat\bSigma - \bSigma) \bOmega\|_{\max} = O_p(\tau_{N,p}),
$$
and then it readily follows that if $\lambda \asymp \tau_{N,p}$,
\begin{align*}
\|\widehat\bOmega - \bOmega\|_{\max} & \le \|\bOmega (\bI_p - \widehat\bSigma\widehat\bOmega)\|_{\max} + \|(\bI_p - \widehat\bSigma \bOmega )' \widehat\bOmega\|_{\max} \\
& \le \|\bOmega\|_1 \|\bI_p - \widehat\bSigma\widehat\bOmega\|_{\max} + \|\bI_p - \widehat\bSigma \bOmega\|_{\max} \|\widehat\bOmega\|_1 \le \lambda \|\bOmega\|_1 + \tau \|\bOmega\|_1 = O_p(\tau_{N,p}),
\end{align*}
where the first term of the last inequality uses the constraint of (\ref{eq2.10}) while the optimality condition of (\ref{eq2.10}) is applied to bound $\|\widehat\bOmega\|_1$ by $\|\bOmega\|_1$. So it remains to find $\tau_{N,p}$ in (\ref{eqa.6}). Since $\bOmega \in \mathcal F(s, C_0)$, $\|\bOmega\|_1 \le C_0$, so we just need to bound $\|\widehat\bSigma - \bSigma_N\|_{\max}$ and $\|\bSigma_N - \bSigma\|_{\max}$. Obviously,
$$
\|\bSigma_N - \bSigma\|_{\max} = O_p\Big(\sqrt{\frac{\log p}{N}}\Big).
$$
By assumption $\|\widehat\bSigma - \bSigma_N\|_{\max} = O_P(a_{m, N, p})$. Thus $\tau = \sqrt{\log p/N} + a_{m, N, p}$. Similar proof as in \cite{CLL11} can also reach error bounds under $\|\cdot\|_1$ and $\|\cdot\|_2$, which we omit. The proof is now complete.
\end{proof}

\section{Technical lemmas}

\begin{lem}\label{tla.PCA1}
 (i)$\|\bLambda'\bU\|_F^2=O_P(np)$, \\
 (ii) $\|\bU'\bU\|_F^2 = O_P(np^2 + pn^2)$, \\
 (iii)$\|\bU'\bU\bF\|_F^2=O_P(np^2 + pn^2)$.
\end{lem}
\begin{proof}
We simply apply Markov inequality to get the rates.
$$
\mathbb E\|\bLambda'\bU\|_F^2 = \mathbb E[ \tr(\bLambda'\bU\bU'\bLambda) ]= n \cdot \tr(\bLambda'\bSigma\bLambda) \le n \|\bSigma\| \cdot \tr(\bLambda'\bLambda) = O(np)\,.
$$
\begin{align*}
\mathbb E\|\bU'\bU\|_F^2  & = \mathbb E\Big[\sum_{t=1}^n \sum_{t'=1}^n (\sum_{j = 1}^p u_{jt} u_{jt'})^2\Big] = \sum_{j_1, j_2 = 1}^p \Big( \sum_{t=1}^n \mathbb E[u_{j_1 t}^2 u_{j_2 t}^2] + \sum_{1\le t \ne t_1 \le n} \sigma_{j_1j_2}^2\Big) \\
& = O_P(np^2 + pn^2)\,,
\end{align*}
since $\sum_{j_1,j_2} \sigma_{j_1j_2}^2 = \tr(\bSigma^2) \le \|\bSigma\|\tr(\bSigma) = O(p)$.
\begin{align*}
\mathbb E\|\bU'\bU\bF\|_F^2 & = \mathbb E\Big[\sum_{t=1}^n \sum_{k=1}^K (\sum_{t' = 1}^n \sum_{j = 1}^p u_{jt} u_{jt'} f_{t'k})^2\Big] \\
& = \sum_{k=1}^K \sum_{j_1, j_2 = 1}^p \Big( \sum_{t=1}^n \mathbb E[u_{j_1 t}^2 u_{j_2 t}^2] f_{tk}^2 + \sum_{1 \le t \ne t_1 \le n} \sigma_{j_1j_2}^2 f_{t_1 k}^2 \Big)= O_P(np^2 + pn^2)\,.
\end{align*}
\end{proof}

\begin{lem}\label{tla.PCA2}
 (i)$\|\bLambda'\bU\|_{\max}=O_P(\sqrt{p\log n})$. \\
 (ii) $\|\bU'\bU\|_{\max} = O_P(p)$, \\
 (iii)$\|\bU'\bU\bF\|_{\max} =O_P(\sqrt{np\log n} + p\sqrt{\log n})$.
\end{lem}
\begin{proof}
(i) $\|\bLambda'\bU\|_{\max} = \max_{t,k} |\bu_t'\blambda_k|$ where $\blambda_k$ is the $k^{th}$ column of $\bLambda$. Since $\bu_t'\blambda_k$ is mean zero sub-Gaussian with variance proxy $\blambda_k'\bSigma\blambda_k \le \|\bSigma\| \|\blambda_k\|^2 = O(p)$, we have $\|\bLambda'\bU\|_{\max} = O_P(\sqrt{p\log n})$.

(ii) $\|\bU'\bU\|_{\max} = \max_{t,t'} |\bu_t'\bu_{t'}| \le \max_{t \ne t'} |\bu_t'\bu_{t'}| + \max_t |\bu_t'\bu_t|$. We need to bound each term separately. The second term is bounded by the upper tail bound of Hanson-Wright inequality for sub-Gaussian vector \citep{HsuKakZha12, RudVer13} i.e.
$$
\mathbb P(\|\bu_t\|^2 > \tr(\bSigma) + 2\sqrt{\tr(\bSigma) s} + 2\|\bSigma\|s) \le e^{-s}\,.
$$
Choose $s = \log n$ and apply union bound, we have $\max_t |\bu_t'\bu_t| = O_P(\tr(\bSigma) + 2\sqrt{\tr(\bSigma) s}) = O_P(p + \sqrt{p \log n}) = O_P(p)$. Then we deal with the first term. By Chernoff bound,
$$
\mathbb P(\max_{t \ne t'} |\bu_t'\bu_{t'}| > s) \le 2n^2 e^{-s\theta} \mathbb E[\exp(\theta \bu_t'\bu_{t'})]\,,
$$
where $\mathbb E[\exp(\theta \bu_t'\bu_{t'})] = E[\exp(\theta^2 \bu_t'\bSigma\bu_t/2)] \le E[\exp(C\theta^2 \|\bu_t\|^2)]$. \cite{HsuKakZha12} showed that
$$
\mathbb E[\exp(\eta \|\bu_t\|^2)] \le \exp\Big( \tr(\bSigma) \eta + \frac{\tr(\bSigma^2) \eta^2}{1-2\|\bSigma\|\eta} \Big)
$$
For $\eta < 1/(4\|\bSigma\|) \le \tr(\bSigma)/(4\tr(\bSigma^2))$, the right hand side is less than $\exp(3\tr(\bSigma)\eta/2) \le \exp(Cp\eta)$. Choose $\eta = C \theta^2$, we have
$$
\mathbb P(\max_{t \ne t'} |\bu_t'\bu_{t'}| > s) \le 2n^2 \exp(-s\theta + C\theta^2 p) \,.
$$
We minimize the right hand side and choose $\theta = s/(2Cp)$, it is easy to check $\eta < 1/(4\|\bSigma\|)$ and see that $\max_{t \ne t'} |\bu_t'\bu_{t'}| = O_P(\sqrt{p \log n})$. So we conclude that $\|\bU'\bU\|_{\max} = O_P(p)$.

(iii) Let $\bar \bff_k$ be the $k^{th}$ column of $\bF$. $\|\bU'\bU\bF\|_{\max} = \max_{t,k} |\bu_t'\bU \bar\bff_k| \le \max_{t,k} |\bu_t'\bU_{(-t)} \bar \bff_{k(-t)}| + \max_{t,k} |\bu_t'\bu_t f_{tk}|$ where $\bU_{(-t)}$, $\bar\bff_{k(-t)}$ are $\bU$ and $\bar\bff_k$ canceling the $t^{th}$ column and element respectively. From (ii) we know the second term is of order $O_P(p\max_{tk}|f_{tk}|) = O_P(p\sqrt{\log n})$. Define $\bxi = \bU_{(-t)} \bar\bff_{k(-t)} \sim \text{subGaussian}({\bf 0}, \bSigma\|\bar\bff_{k(-t)}\|^2)$, which is independent with $\bu_t$. Thus
$$
\mathbb P(\max_{t,k}|\bu_t'\bxi| > s) \le 2nK e^{-s\theta} \mathbb E[\exp(\theta\bu_t'\bxi)]\,,
$$
where $\mathbb E[\exp(\theta\bu_t'\bxi)] \le \mathbb E[\exp(\theta^2 \bu_t'\bSigma\bu_t \|\bar\bff_{k(-t)}\|^2/2)] \le \mathbb E[\exp(C\theta^2 n \|\bu_t\|^2)]$. Similar to (ii), we choose $\eta = C\theta^2 n$ here. It is not hard to see $\max_{t,k}|\bu_t'\bxi| = O_P(\sqrt{np\log n})$. Thus $\|\bU'\bU\bF\|_{\max} = O_P(\sqrt{np\log n} + p\sqrt{\log n})$.
\end{proof}

\begin{lem}\label{tla.1}
 (i)$\|\bF'\bU'\|_F^2=O_P(np)$. \\
 (ii) $\|\bU'\Phi(\bW)\|_F^2=O_P(npJ)$, $\| \bU' \Phi(\bW)\bB\|_F^2=O_P(np)$.\\
 (iii) $\|\Phi(\bW)'\bU\bF\|_F^2=O_P(npJ)$, $\|\bB'\Phi(\bW)'\bU\bF\|_F^2=O_P(np)$.
\end{lem}
\begin{proof}
This results can be found in the paper of Fan, Liao and Wang (2014). But the conditions they used are a little bit different from our conditions. In particular, we allow no time (sample) dependence and only require bounded $\|\bSigma\|_2$ instead of $\|\bSigma\|_1$. By Markov inequality, it is sufficient to show the expected value of each term attains the corresponding rate of convergence. \\
$$
\mathbb E\|\bF'\bU'\|_F^2 = \mathbb E[ \tr(\bF' \mathbb E[\bU'\bU] \bF) ]  = \mathbb E[ \tr(\bF' \tr(\bSigma)\bF) ] = n \cdot \tr(\bSigma) = O(np).
$$
\begin{align*}
\mathbb E\|\bU'\Phi(\bW)\|_F^2 & = \mathbb E[ \tr(\Phi' \mathbb E[\bU\bU'|\bW] \Phi) ] =  n \cdot \mathbb E[ \tr(\Phi' \bSigma \Phi) ] \le nJd \cdot \mathbb E[ \|\Phi' \bSigma \Phi\|_2 ] \\
& \le nJdC_0 \mathbb E[ \|\Phi' \Phi\|_2 ] = O(npJ).
\end{align*}
$$
\mathbb E \|\Phi(\bW)'\bU\bF\|_F^2 = \mathbb E[\tr(\Phi'\mathbb E[\bU\bF\bF'\bU'|\bW]\Phi)] = \mathbb E[\tr(\bF\bF') \tr(\Phi'\bSigma\Phi)] = O(npJ).
$$
$\mathbb E\|\bU'\Phi(\bW)\bB\|_F^2$ and $\|\bB'\Phi(\bW)'\bU\bF\|_F^2$ are both $O(np)$ following the same proof as above. Thus the proof is complete.
\end{proof}

\begin{lem}\label{tla.2}
 (i) $\|\bF'\bU'\|_{\max} = O_P(\sqrt{n\log p} )$ \\
 (ii) $\|\bU'\Phi(\bW)\|_{\max} = O_P(\phi_{\max}\sqrt{p \log (nJ) })$, $\| \bU' \Phi(\bW)\bB\|_{\max}=O_P(\sqrt{p \log n})$.\\
 (iii) $\|\Phi(\bW)'\bU\bF\|_{\max} = O_P(\phi_{\max}\sqrt{np\log J} )$, $\|\bB'\Phi(\bW)'\bU\bF\|_{\max}=O_P(\sqrt{np})$.
\end{lem}

\begin{proof}
 (i) It is not hard to see $\|\bF'\bU'\|_{\max} = \max_{k \le K, j \le p} |\sum_{t = 1}^n f_{tk} u_{jt} | = O_p(\sqrt{n \log p}).$ The detailed proof by Chernoff bound is given in the following. By union bound and Chernoff bound, we have
$$
\mathbb P\Big( \max_{k \le K, i \le p} \Big|\sum_{t = 1}^n f_{tk} u_{jt} \Big| > t\Big) \le 2pK e^{-t\theta} \cdot \mathbb E\Big[e^{\theta \sum_{t=1}^n f_{tk} u_{jt}}\Big].
$$
The expectation is calculated by fist conditioning on $\bF$,
$$
E\Big[e^{\theta \sum_{t=1}^n f_{tk} u_{jt}}\Big] = \mathbb E\Big[ \mathbb E\Big[e^{\theta \sum_{t=1}^n f_{tk} u_{jt}} | \bF\Big] \Big] \le  \mathbb E\Big[ e^{\theta^2 \sum_{t=1}^n f_{tk}^2 \sigma_{jj} /2} \Big] \le e^{\frac12 n C_0 \theta^2},
$$
where the second equality uses the sub-Gaussianity of $u_{jt}$ and the last inequality is from $n^{-1} \bF'\bF = \bI$ and $\|\bSigma\|_2 \le C_0$.
Therefore, choosing $\theta = \frac{t}{nC_0}$, we have
$$
\mathbb P\Big( \max_{k \le K, j \le p} \Big| \sum_{t = 1}^n f_{tk} u_{jt} \Big|> t\Big) \le 2pK e^{-t\theta} e^{\frac{C_0}{2} n\theta^2 } = 2pK e^{-\frac{t^2 }{2 C_0 n}}.
$$
Thus $\|\bF'\bU'\|_{\max} = O_p(\sqrt{n \log p})$.

(ii) $\|\bU'\Phi(\bW)\|_{\max} = \max_{\nu, l, t} |\sum_{j=1}^p u_{jt} \phi_{\nu }(W_{jl})| =  \max_{\nu, l, t} |\bar\phi_{\nu l}' \bu_{t}|$, where $\bar\phi_{\nu l} = (\phi_{\nu}(\bW_{1l}), \dots, \phi_{\nu}(\bW_{pl}))'$. Consider the tail probability condition on $\bW$:
$$
\mathbb P\Big(\max_{\nu \le J, l \le d, k \le n} |\bar\phi_{\nu l }' \bu_{k}| > t \Big| \bW \Big) \le 2Jdn \cdot e^{-t\theta} \mathbb E[e^{\theta \bar\phi_{\nu l }' \bu_{k}} | \bW] \le 2Jdn \cdot \exp\Big\{-t\theta + \frac 12 \theta^2 \bar\phi_{\nu l }' \bSigma \bar\phi_{\nu l }\Big\}.
$$
The right hand side can be further bounded by
$$
2Jdn \cdot \exp\Big(-t\theta + \frac 12 \theta^2 C_0 \|\bar\phi_{\nu l }\|^2\Big) \le 2Jdn \cdot \exp\Big(-t\theta +  \frac 12 p C_0 \theta^2 \phi_{\max}^2 \Big).
$$
Choose $\theta$ to minimize the upper bound and take expectation with respect to $\bW$, we obtain
$$
\mathbb P\Big(\max_{\nu \le J, l \le d, k \le n} |\bar\phi_{\nu l }' \bu_{k}| > t  \Big) \le 2Jdn \cdot \exp\Big\{ -\frac{t^2}{2pC_0 \phi_{\max}^2}\Big\}.
$$
Finally choose $t \asymp \phi_{\max} \sqrt{p \log(nJ)} $, the tail probability is arbitrarily small with a proper constant. So $\|\bU'\Phi(\bW)\|_{\max} = O_P(\phi_{\max}\sqrt{p\log (nJ)}  )$. The second part of the results follows similarly. Note $\| \bU' \Phi(\bW)\bB\|_{\max} \le \| \bU' \bG(\bW)\|_{\max} + \|\bU' \bR(\bW)\|_{\max}$ and the first term dominates. So the same derivation gives
$$
\mathbb P\Big( \| \bU' \bG(\bW)\|_{\max} > t\Big) \le 2Kn \cdot \exp \Big\{-\frac{t^2}{2 C_0\|\bar g_k\|^2}\Big\},
$$
where $\bar g_k = (g_k(\bW_1), \dots, g_k(\bW_p))$. $\|\bar g_k\|^2 = O_p(p)$ since it is
assumed eigenvalues of $p^{-1} \bG(\bW)'\bG(\bW)$ is bounded almost surely. Hence, $\| \bU' \Phi(\bW)\bB\|_{\max}=O_P(\sqrt{p\log n} )$.

(iii) $\|\Phi(\bW)'\bU\bF\|_{\max} = \max_{\nu \le J, l \le d, k \le K} |\sum_{j=1}^p \sum_{i = 1}^n \phi_{\nu}(W_{jl})u_{ji} f_{ik}|$. Using Chernoff bound again, we get
$$
\mathbb P\Big( \max_{\nu \le J, l \le d, k \le K} \Big|\sum_{j=1}^p \sum_{i = 1}^n \phi_{\nu}(W_{jl})u_{ji} f_{ik}\Big| > t\Big) \le 2JdK\cdot e^{-t\theta} \cdot \mathbb E\Big[e^{\theta \sum_{t=1}^n f_{tk} \bar\phi_{\nu l}' \bu_{t} }\Big] .
$$
Since $\sum_{t=1}^n f_{tk} \bar\phi_{\nu l}' \bu_{t} | \bF \sim \text{sub-Gaussian}(0, \sum_{t=1}^n f_{tk}^2 \bar\phi_{\nu l}' \bSigma \bar\phi_{\nu l}) = \text{sub-Gaussian}(0, n \bar\phi_{\nu l}' \bSigma \bar\phi_{\nu l})$, the right hand side is easy to bound by first conditioning on $\bF$.
$$
\mathbb E\Big[e^{\theta \sum_{t=1}^n f_{tk} \bar\phi_{\nu l}' \bu_{t} }\Big] \le \mathbb E\Big[ \exp \Big( \frac 12 n \theta^2 \bar\phi_{\nu l}' \bSigma \bar\phi_{\nu l} \Big) \Big] \le E\Big[ \exp \Big( \frac 12  npC_0 \phi_{\max}^2 \theta^2 \Big) \Big].
$$
Therefore, choosing $\theta = \frac{t}{npC_0\phi_{\max}^2}$, we have
$$
\mathbb P\Big( \|\Phi(\bW)'\bU\bF\|_{\max} > t \Big) \le 2JdK \cdot \exp\Big\{-t\theta +
\frac{1}{2} npC_0 \phi_{\max}^2\theta^2\Big\} = 2JdK \exp\Big\{-\frac{t^2}{2npC_0\phi_{\max}^2}\Big\}.
$$
So we conclude $ \|\Phi(\bW)'\bU\bF\|_{\max}  = O_p(\phi_{\max} \sqrt{np\log J})$. By similar derivation as in (ii), we also have $\|\bB'\Phi(\bW)'\bU\bF\|_{\max}$ and $\|\bG(\bW)'\bU\bF\|_{\max}$ are both of order $O_P(\sqrt{np})$.

\end{proof}

\begin{lem}\label{tla.5}
(i) $\|\bU\bU'\bLambda\|_{\max}=O_P(\sqrt{np\log p}+n\|\bSigma\|_1)$, \\
(ii) $\|\bU\bU'\Phi(\bW)\|_{\max}=O_P(\phi_{\max}(\sqrt{np\log p} + n\|\bSigma\|_1))$ and $\|\bU\bU'\Phi(\bW)\bB\|_{\max}=O_P(\sqrt{np\log p} + nJ\phi_{\max} \|\bSigma\|_1)$.
\end{lem}
\begin{proof}
(i) $\|\bU\bU'\bLambda\|_{\max} \le \max_{j,k} |\sum_{t = 1}^n u_{jt} \bu_t'\blambda_k - n \sum_{j' = 1}^p \sigma_{jj'} \lambda_{j'k}| + n \max_{j,k} \sum_{j' = 1}^p |\sigma_{jj'}| | \lambda_{j'k}|$. The second term is $O(n\|\bSigma\|_1)$. So it suffices to focus on the first term. Let $\bSigma = \bA\bA'$ and $\bu_t = \bA \bv_t$ so that $\Var(\bv_t) = \bI$. Write $\bA' = (\ba_1,\dots, \ba_p)$, so we have $u_{jt} = \ba_j'\bv_t$. Also denote $\bd_k = \bA' \blambda_k$. Thus $u_{jt} \bu_t'\blambda_k = \ba_j'\bv_t \bv_t' \bd_k$ and $\sum_{j' = 1}^p \sigma_{jj'} \lambda_{j'k} = \ba_j'\bd_k$.
\begin{equation} \label{eqF.1}
\mathbb P\Big( \max_{j,k} \Big|\sum_{t = 1}^n  (\ba_j'\bv_t \bv_t' \bd_k - \ba_j'\bd_k)\Big| > s \Big) \le pK\mathbb P\Big(\Big|\sum_{t = 1}^n  (\widetilde\ba_j'\bv_t \bv_t' \widetilde\bd_k - \widetilde\ba_j'\widetilde\bd_k)\Big| > \frac{s}{\max_{j,k}\|\ba_j\|\|\bd_k\|} \Big)\,,
\end{equation}
where $\widetilde\ba_j$ and $\widetilde\bd_k$ are two unit vectors of dimension $p$. We will bound the right hand side with arbitrary unit vectors $\widetilde\ba_j$ and $\widetilde\bd_k$.
\begin{align*}
& \mathbb P \Big(\Big|\sum_{t = 1}^n  \widetilde\ba_j'\bv_t \bv_t' \widetilde\bd_k - n\widetilde\ba_j'\widetilde\bd_k\Big| > s \Big) \\
& \le \mathbb P\Big(\Big|\sum_{t = 1}^n  ((\widetilde\ba_j + \widetilde\bd_k)'\bv_t)^2 - n \|\widetilde\ba_j + \widetilde\bd_k\|^2\Big| > 2s \Big) + \mathbb P\Big(\Big|\sum_{t = 1}^n  ((\widetilde\ba_j - \widetilde\bd_k)'\bv_t)^2 - n \|\widetilde\ba_j - \widetilde\bd_k\|^2\Big| > 2s \Big)\,.
\end{align*}
Note that $(\widetilde\ba_j + \widetilde\bd_k)'\bv_t \sim \text{subGaussian}(0, \|\widetilde\ba_j + \widetilde\bd_k\|^2)$ and $\|\widetilde\ba_j + \widetilde\bd_k\|^2 \le 4$. By Bernstein inequality, we have for constant $C > 0$,
$$
\mathbb P \Big(\Big|\sum_{t = 1}^n  (\widetilde\ba_j'\bv_t \bv_t' \widetilde\bd_k - \widetilde\ba_j'\widetilde\bd_k)\Big| > s \Big) \le 2\exp\Big( -C \min(s^2/n, s)\Big)\,.
$$
Choose $s = C \sqrt{n\log p} \max_{jk} \|\ba_j\|\|\bd_k\|$ in (\ref{eqF.1}), we can easily show that the exception probability is small as long as $C$ is large enough. Therefore, noting $\max_{jk} \|\ba_j\|\|\bd_k\| \le C\max_{k} \|\blambda_k\|$, $\max_{j,k} |\sum_{t = 1}^n u_{jt} \bu_t'\blambda_k  - n \sum_{j' = 1}^p \sigma_{jj'} \lambda_{j'k}| =  O_P( \sqrt{n\log p} \max_{k} \|\blambda_k\|) = O_P(\sqrt{np\log p})$. Finally $\|\bU\bU'\bLambda\|_{\max} = O_P(\sqrt{np\log p}+n\|\bSigma\|_1)$.

(ii) The rates of $\|\bU\bU'\Phi(\bW)\|_{\max}$ and $\|\bU\bU'\Phi(\bW)\bB\|_{\max}$ can be similarly derived as (i). Denote $\bPhi_{vl} = (\phi_v(W_{1l}), \dots, \phi_v(W_{pl}))'$, so
\begin{align*}
\|\bU\bU'\Phi(\bW)\|_{\max} &\le \max_{j,v,l} \Big|\sum_{t = 1}^n u_{jt} \bu_t'\bPhi_{vl} - n \sum_{j' = 1}^p \sigma_{jj'} \phi_v(W_{j'l})\Big| + n \max_{j,v,l} \sum_{j' = 1}^p |\sigma_{jj'}| |\phi_v(W_{j'l})| \\
& = O_P(\sqrt{n\log p} \max_{v,l} \|\bPhi_{vl}\| + n\phi_{\max} \|\bSigma\|_1) = O_P(\phi_{\max}(\sqrt{np\log p} + n\|\bSigma\|_1))\,.
\end{align*}
Denote the $k^{th}$ column of $\Phi(\bW)\bB$ by $(\bPhi\bB)_k$, we have
\begin{align*}
\|\bU\bU'\Phi(\bW)\bB\|_{\max}  &\le \max_{j,k} \Big|\sum_{t = 1}^n u_{jt} \bu_t'(\bPhi\bB)_k - n \sum_{j' = 1}^p \sigma_{jj'} (\bPhi\bB)_{j'k})\Big| + n \max_{j,k} \sum_{j' = 1}^p |\sigma_{jj'}| |(\bPhi\bB)_{j'k}| \\
& = O_P(\sqrt{n\log p} \max_{k} \|(\bPhi\bB)_k\| + nJ\phi_{\max} \|\bSigma\|_1) = O_P(\sqrt{np\log p} + nJ\phi_{\max} \|\bSigma\|_1)\,,
\end{align*}
where we use $\max_{k} \|(\bPhi\bB)_k\| \le \|\bPhi\bB\|_F = O_P(\sqrt{p})$.
\end{proof}

\end{document}